\documentclass[11pt]{article}
\usepackage{libertine}
\usepackage{fullpage}
\usepackage{microtype}
\usepackage{subcaption}

\usepackage{amsmath,amssymb,amsthm,mathtools,thmtools,thm-restate}

\usepackage[breaklinks=true]{hyperref}
\usepackage[svgnames]{xcolor}
\usepackage[capitalise,nameinlink]{cleveref}
\hypersetup{colorlinks={true},linkcolor={DarkBlue},citecolor=[named]{DarkGreen}}

\usepackage{natbib}

\usepackage{algorithm}
\usepackage{algpseudocode}
\usepackage{float}
\usepackage[shortlabels]{enumitem}
\usepackage{tikz}
\usetikzlibrary{backgrounds,patterns}
\usepackage{array}

\newtheorem{theorem}{Theorem}
\newtheorem{lemma}{Lemma}

\newtheorem{observation}{Observation}

\theoremstyle{definition}

\usepackage[svgnames]{xcolor}

\usepackage{cleveref}

\newcommand{\MMS}{\text{MMS}}

\usepackage{authblk}

\title{\bf Maximin share allocations for two types of items}

\author[1]{Yu Chen}
\affil[1]{National University of Singapore}
\author[2]{Kai Jie Wong\thanks{Corresponding author. Email: kaijie.wong@rutgers.edu. This work was conducted while the author was an undergraduate at the National University of Singapore.}}
\affil[2]{Rutgers University}
\author[1]{Xiaorui Zhang}

\date{}

\begin{document}
\maketitle

\begin{abstract}
We study the fair allocation of indivisible items among \(n\) agents with additive valuations (or costs) using the fairness notion of maximin share (MMS). Since MMS allocations do not always exist when \(n > 2,\) one direction in the literature has been to identify special classes of instances in which exact MMS allocations can be guaranteed. We show that, when items can be divided into two types of chores, MMS allocations always exist and can be computed in polynomial time using a novel geometric approach inspired by prior work in job scheduling. We further discuss how our technique extends to instances with two types of goods, as well as mixed instances consisting of one type of objective good and one type of objective chore.
\end{abstract}

\section{Introduction}

The fair allocation of indivisible items among agents with differing preferences arises naturally in many settings and has attracted significant interest across computer science and multi-agent systems \citep{aziz_developments_2020}; see also recent surveys by \cite{10.1016/j.artint.2023.103965, liu2024mixed}. The goal is to find a \textit{fair} partition of a set \(M\) of \(m\) items among a set \(N\) of \(n\) agents, where each agent \(i\in N\) has some valuation function \(v_i(\cdot)\) over subsets of \(M.\) In this work, we assume that agents have additive preferences, such that the utility of a bundle \(S\) is given by \(v_i(S)=\sum_{o\in S}v_i(o),\) where \(v_i(o)\) represents the utility of item \(o.\) An item \(o\) can be considered to be a \textit{good} (in which case \(v_i(o)>0,\) or a \textit{chore} (\(v_i(o)<0\)). 

One natural and widely-studied fairness notion in this setting is the \textit{maximin share} (MMS) value. Introduced by \cite{budish2011} in the context of course allocation, an agent's MMS value is the maximum utility they can guarantee for themselves, if they partition the items into \(n\) bundles, then receive the worst one. Formally, let \(\Pi\) denote the set of all partitions of \(M\) into \(n\) bundles. Let \((X_i)_{i\in N}\in \Pi\) be any such partition. Then, we write the \(\MMS\) value for agent \(i\in N\) as
\[
\MMS_i:= \max_{(X_i)_{i\in N}\in \Pi}\min _{j\in N}v_i(X_j).
\]
An allocation is said to be MMS if every agent receives a bundle worth at least their MMS value. 

The MMS is a compelling benchmark for fairness, with evidence from real-world experiments indicating that participants tend to prefer it \citep{gates_helpful_2020}. However, despite its intuitive appeal, exact MMS allocations may not always exist, even when there are only \(n\geq 3\) agents with additive valuations \citep{kurokawa_fair_2018,feige_tight_2021}.

In view of the difficulty of attaining MMS allocations for general additive valuations, we turn our attention to the special case when there are at most two types of items. We assume that the item set \(M\) can be partitioned into two ``types" \(A\) and \(B.\) Items of the same type are identical, so a fixed agent will derive the same utility from items of the same type. A motivating example, in the case of chores, could be a group of housemates fairly splitting up \(|A|\) cooking chores, and \(|B|\) cleaning chores.  Another plausible scenario is to divide chores according to when they are performed: the chores in \(A\) are assigned for this week, while those in \(B\) are assigned for the following week. Since housemates may have different preferences over when they would rather do chores, assigning a chore to one week rather than the other can impose different levels of disutility on different housemates.

At the same time, we consider the more general case of mixed instances, where items may be either goods or chores. For our setting, we assume that the items are either \textit{objective goods} (which yield positive utility for all agents) or \textit{objective chores} (which yield negative utility for all agents), as formalized by Aziz et al. \cite{aziz_fair_2021}.

\subsection{Contributions}

We give a polynomial-time algorithm for computing MMS allocations for two-item-type instances in each of the following settings:
\begin{itemize}
    \item at most two types of chores
    \item at most two types of goods
    \item exactly one type of \textit{objective good} and one type of \textit{objective chore} (a ``mixed" instance).
\end{itemize}

In the case of chore allocation, it is well-known that when agents are identical, the MMS allocation problem is equivalent to the classic makespan minimization in job scheduling \citep{huang_reduction_2023}. Building on this connection, our approach is based on a geometric perspective, inspired by prior work in job scheduling when there are only two types of jobs \citep{mccormick_polynomial_1997}. 

A similar geometric viewpoint was also used by \cite{gorantla_fair_2023} to obtain EFX (envy-free up to any good) allocations for two types of goods. However, their techniques differ significantly from ours due to the different nature of EFX and MMS.

Since there are few settings in which MMS allocations are known to exist, these results provide new positive cases for exact MMS computability. A notable feature of our approach is that essentially the same underlying technique applies across goods, chores, and mixed instances.

We first present the algorithm setting proof for the chore setting in Sections~\ref{sec:overview} and~\ref{section:main-result}. In Section~\ref{section:Extension}, we show how the same ideas extend naturally to goods and mixed instances using analogous arguments.

\paragraph{Independent work.} The two-item-types setting can also be viewed as a natural subclass of \textit{personalised bivalued} instances\footnote{An instance is \textit{personalised bivalued} if \(v_i(o)\in \{a_i,b_i\}\) for all \(o\in M.\)}. In the development of our work, the best-known approximation factor for chores was \(\frac{15}{13}\), due to \cite{garg_improved_2025} and existence remained unknown. Independently of our work, \citet{zhang2026exactmmsallocationspersonalized} recently proved the existence and polynomial-time computability of exact MMS allocations for personalized bivalued valuations, both when all items are goods and when all items are chores. These results generalize our corresponding results for goods-only and chores-only instances. To the best of our knowledge, Zhang's results do not address mixed instances containing both objective goods and objective chores, which we handle in our mixed-items result.

\subsection{Further related work}

\paragraph{Few types of items.} Several works have studied this setting in the case of envy-based fairness, under additive valuations. \cite{aziz_fair_2023} provide a polynomial-time algorithm for EFX allocations when there are only two types of chores. \cite{gorantla_fair_2023} give a polynomial-time algorithm for EFX allocations where there are only two types of goods, as well as existence of envy-free allocations for arbitrary number of types of goods, provided each type appears in sufficiently large quantity. 

\paragraph{Other structured instances.} For the case of bivalued instances\footnote{An instance is \textit{bivalued} if \(v_i(o)\in \{a,b\}\) for all \(o\in M.\)}, MMS allocations are known to exist \cite{feige_maximin_2022}, for both goods and chores. MMS allocations are also possible under ternary valuations (agent values in \(\{0,1,2\}\)) for goods \citep{amanatidis2017}, factored instances for chores \citep{garg_improved_2025}, and weakly lexicographic valuations for both goods and chores \citep{ebadian2022}. There is also a line of work showing how MMS allocations can be computed in polynomial time when agents have binary submodular valuations for goods \cite{barman_existence_2021,viswanathan_general_2025}, or binary supermodular costs for chores \cite{barman_fair_2023}.

\paragraph{Approximations.} Since MMS allocations are not guaranteed in general additive instances, the study of approximate MMS allocations has also seen significant progress. For goods, the first $2/3$-MMS guarantee by \cite{kurokawa_fair_2018} was improved to $3/4$ by \cite{journals/mor/GhodsiHSSY21}. The first constant factor improvement to $3/4+3/3836$ was made by \cite{conf/soda/AkramiG24}, followed by $10/13$ \citep{conf/soda/HeidariKSS26}, and most recently $7/9$ by \cite{huang2025fptas79approximationmaximinshare}.

For chores, \cite{aziz2017} show a $4/3$-approximation in polynomial time. The existence guarantee was subsequently improved to $11/9$ \citep{huang_algorithmic_2021}, and finally $13/11$, with a $(13/11+\epsilon)$-approximation in polynomial time \citep{huang_reduction_2023} via an algorithmic reduction to job scheduling.

\section{Preliminaries}\label{section:Preliminaries}

Let \(M\) be a set of \(m\) indivisible chores, and \(N\) be a set of \(n\) agents. In this section and Section~\ref{section:main-result}, we focus on chore allocation. Since all items yield negative utility, it is more convenient to work with positive disutilities instead. Accordingly, for each agent \(i\in N\) we let \(d_i:M\to \mathbb{R}_{> 0}\) denote an additive disutility function. We consider the setting where the chore set can be partitioned as \(M=A\cup B\) where for each \(i\in N\), if \(c\in A\) then \(d_i(c):=\alpha_i\) for some \(\alpha_i\in\mathbb{R}_{> 0},\) and if \(c\in B\) then \(d_i(c):=\beta_i\) for some \(\beta_i\in\mathbb{R}_{> 0}.\) If \(c\in A,\) we say that \(c\) is a Type A chore; otherwise it is a Type B chore. Since chores of the same type are identical, any two-chore-type fair division instance can be succinctly described by  \(\mathcal{I}:=(a,b,N)\) where \(a,b\) are the number of Type A and Type B chores respectively. 

\paragraph{Maximin share.} Recall \(\Pi\) is the set of all partitions of \(M\) into \(n\) bundles. For an instance of chore allocation, we instead define the MMS share as a positive value
\[
\MMS_i(M) := \min_{{(X_i)_{i\in N}} \in \Pi} \max_{j \in N} d_i(X_j),
\]
thus an allocation is said to MMS if every agent receives a bundle of disutility at most their MMS value. We may also write \(\MMS_i(\mathcal{I})\) to mean the MMS value for agent \(i\) given the instance \(\mathcal{I},\) and just \(\MMS_i\) when the instance is obvious. We say that a bundle \(B\) is \textit{MMS-feasible} for agent \(i\) if \(d_i(B)\leq \MMS_i.\) 

\section{Technique overview}\label{sec:overview}

Our approach is based on a geometric view of two-type allocations, building on the polynomial-time algorithm of~\cite{mccormick_polynomial_1997} for scheduling with only two job types. Given a makespan threshold \(T,\) their algorithm determines in polynomial time whether there exists a schedule with makespan at most \(T\), and, if so, outputs such a schedule. Consequently, binary search over \(T\) yields the minimum makespan in polynomial time. For chore allocation with identical agents, this is equivalent to computing an MMS allocation:

\begin{lemma}[\cite{mccormick_polynomial_1997}]\label{lemma:mms-identical}
    When agents are identical and there are only two types of chores, an MMS allocation can be computed in polynomial time.
\end{lemma}

 With heterogeneous agents, however, each agent now has a different disutility function and hence different MMS thresholds. This prevents the scheduling algorithm from being applied as is. We instead leverage key geometric lemmas underlying their algorithm, together with their algorithm as a subroutine to handle heterogeneous agents. Our procedure goes beyond a direct adaptation of the scheduling techniques, as making use of them for heterogeneous agents is nontrivial. We first introduce the geometric formulation. 

\paragraph{Geometric formulation.} When there are only two types of chores, any bundle \(S\) assigned to an agent can be represented by an integral pair \((x,y)\in\mathbb{Z}_{\geq 0}^2\) where \(x\) is the number of Type A chores and \(y\) is the number of Type B so that the disutility of the bundle for agent \(i\) is \(d_i(S)=\alpha_ix+\beta_iy\). Observe that the points \((x_1,y_1),(x_2,y_2),\dots,(x_n,y_n)
\in\mathbb{Z}_{\geq 0}^2\) correspond to a complete allocation if \(\sum_{i=1}^n(x_i,y_i)=(a,b),\) i.e. the allocation consists of nonnegative integral bundles which uses up all \(a\) Type A chores and all \(b\) Type \(B\) chores. From now on we may use the terms ``integer point" and ``bundle" interchangeably. 

For each \(i\in N,\) define \(L_i\) to be the line \(\alpha_ix+\beta_iy=\MMS_i\). The set of (possibly fractional) MMS-feasible bundles for agent \(i\) is exactly the nonnegative region \(R_i\) defined by
\[
R_i:=\{(x,y)\in \mathbb{R}^2_{\geq 0}: \alpha_ix+\beta_iy\leq \MMS_i\},
\]
which is also known as the \textit{fractional knapsack polytope}. Let the convex hull of integral points (the ``integer hull") in \(R_i\) be \(IH_i:=\text{conv}(R_i\cap\mathbb{Z}^2).\) This is the corresponding \textit{knapsack polytope}. See Fig~\ref{fig:obs-P}, where the red shaded region is \(IH_1,\) the integer hull of \(R_1,\) which is the nonnegative region under \(L_1.\)

\paragraph{Feasibility point.} Given a two-chore-type instance \(\mathcal{I}=(a,b,N)\), we call the point \(P_\mathcal{I}:=(\frac{a}{n},\frac{b}{n})\) the \textit{feasibility point} of the instance. We write `\(+\)' as regular vector addition of points, and for any point \(T,\) we write \(x(T)\) and \(y(T)\) to be its \(x\)- and \(y\)-coordinates respectively. We have the following useful observation.

\begin{restatable}{observation}{obsP}\label{obs:P}
    \(P_\mathcal{I}=(\frac{a}{n},\frac{b}{n})\in \cap_{i=1}^n IH_i\) for each \(i\in N.\)
\end{restatable}

\begin{figure}[htbp]
    \centering
    \begin{minipage}{0.48\linewidth}
        \centering
        \resizebox{0.8\linewidth}{!}{
        \begin{tikzpicture}[scale=2.5, >=stealth]
            \draw[->] (0,0) -- (2.2,0) node[right] {$x$};
            \draw[->] (0,0) -- (0,2.2) node[above] {$y$};

            \draw[red, thick] (0,2) -- (1.2,0)
                node[pos=0.1, right, red] {$L_1$};
            \draw[green!70!black, thick] (0,1.5) -- (1.8,0)
                node[pos=0.5, right, green!70!black] {$L_2$};
            \draw[blue, thick] (0,1) -- (2,0)
                node[pos=0.8, above right, blue] {$L_3$};

            \fill[black] (0.8,0.5) circle (0.6pt)
                node[left, black] {\(P_\mathcal{I}\)};
        \end{tikzpicture}}
        \caption{Illustration of Observation~\ref{obs:P} with \(n=3\). Here \(P_\mathcal{I}\in\bigcap_{i=1}^3 IH_i\).}
        \label{fig:obs-P}
    \end{minipage}
    \hfill
    \begin{minipage}{0.48\linewidth}
        \centering
        \resizebox{0.85\linewidth}{!}{
        \begin{tikzpicture}[scale=2.5, >=stealth]
            \draw[->] (0,0) -- (2.2,0) node[right] {$x$};
            \draw[->] (0,0) -- (0,2.2) node[above] {$y$};

            \draw[red, thick] (0,2) -- (1.2,0) node[right, pos=0.4] {\(L_1\)};

            \def\x{0.375}
            \def\y{0.375}
            \foreach \i in {0,...,5} {
                \foreach \j in {0,...,5} {
                    \fill[black] ({\i*\x}, {\j*\y}) circle (0.6pt);
                }
            }

            \begin{scope}[on background layer]
                \fill[red!20] (0,\y*5) -- (\x*2,\y*2)
                    -- (\x*3,0) -- (0,0) -- cycle;
            \end{scope}

            \draw (0,5*\y) -- (2*\x,2*\x) -- (3*\x,0) -- cycle;

            \coordinate (P) at (1.75*\x,2.25*\y);
            \fill[red] (P) circle (0.4pt);
        \end{tikzpicture}}
        \caption{Illustration of Lemma~\ref{lemma:IDP}. \(L_1\) is \(5x+3y=16\). The shaded region is \(IH_1=\operatorname{conv}(R_1\cap\mathbb{Z}^2)\).}
        \label{fig:lemma-IDP}
    \end{minipage}
\end{figure}

See Fig~\ref{fig:obs-P} for an illustration of Observation~\ref{obs:P} when \(n=3.\) Observation~\ref{obs:P} states that \(P_\mathcal{I}\) must lie within \(IH_1,IH_2\) and \(IH_3.\) Hence, geometrically it must lie under each of the lines \(L_1,L_2\) and \(L_3.\) 

\begin{proof}[Proof of Observation~\ref{obs:P}]
Observe that given instance \(\mathcal{I}\), any MMS allocation for agent \(i\in N\) corresponds to \(n\) points \((x_1,y_1),\dots,(x_n,y_n)\in R_i\cap \mathbb{Z}^2\) such that \(\sum_{i=1}^n(x_i,y_i)=(a,b).\) Then, by convexity, their centroid \(P\) lies in \(IH_i.\)
\end{proof}

\paragraph{Unimodular triangle allocations.} Next, we adapt results from~\cite{mccormick_polynomial_1997}, which characterize the feasibility of two-job-type job scheduling allocations based on the feasibility point. We restate the lemma in our setting (for completeness, we also include a proof in Appendix~\ref{sec:appendix-IDP}).
\begin{restatable}{lemma}{lemmaIDP}[Corollary 3.6 of \cite{mccormick_polynomial_1997}] \label{lemma:IDP}
    For any two-chore-type instance \(\mathcal{I}=(a,b,N),\) and for each \(i\in N,\) the following statements are equivalent.
    \begin{enumerate}[label=(\roman*)]
        \item \(P_\mathcal{I}\in IH_i\) 
        \item \(P_\mathcal{I}\) lies in some \textit{unimodular} (area \(1/2\)) triangle with vertices in \(R_i\cap \mathbb{Z}^2\)
        \item There is a triplet of integral points \((x_1,y_1),(x_2,y_2),(x_3,y_3)\in R_i\cap \mathbb{Z}^2\) and integral multiplicities \(n_1,n_2,n_3 \in \mathbb{Z}_{\geq 0}\) such that \(n_1(x_1,y_1)+n_2(x_2,y_2)+n_3(x_3,y_3)=(a,b)\) and \(n_1+n_2+n_3=n.\)
    \end{enumerate}
\end{restatable}

See Fig~\ref{fig:lemma-IDP} for an illustration. Here \(a=7,b=9,n=4,\) so the red point \(P_\mathcal{I}=(7/4,9/4).\) The agent (agent 1) has a disutility of 5 for the Type A chores (corresponding to the \(x\)-axis), and 3 for the Type B chores (corresponding to the \(y\)-axis), and their MMS value is 16, so \(L_1\) is \(5x+3y=16.\) Observe that \((7/4,9/4)=\frac{1}{4}[(0,5)+2\cdot (2,2)+(4,0)],\) so \(P_\mathcal{I}\) is a weighted centroid of the triangle with vertices \((0,5),(2,2),(4,0)\in IH_1\) with multiplicities \(1,2,1.\) This corresponds to an MMS allocation for agent 1.

Observation~\ref{obs:P} along with Lemma~\ref{lemma:IDP}  implies that for any two-chore-type instance \(\mathcal{I}\), there exists an MMS-feasible allocation for each agent \(i\) which uses at most three different types of bundles of integral multiplicities \(n_1+n_2+n_3=n\). These bundles form a unimodular triangle with weighted centroid \(P_\mathcal{I}\). This triangle can be found in polynomial time:

\begin{lemma}[Lemma 3.3 of \cite{mccormick_polynomial_1997}] \label{lemma:alg-time}
Given a two-chore-type instance \(\mathcal{I},\) and an agent \(i,\) a unimodular triangle with vertices in \(R_i\cap \mathbb{Z}^2\) containing \(P_\mathcal{I}\) exists and can be found in \(O(\log^2 t)\) time.
\end{lemma}

\section{MMS allocation for two types of chores} \label{section:main-result}

In this section, we describe the allocation procedure for chores and prove the following theorem. 

\begin{theorem}\label{thm:main}
    Given a two-chore-type instance \(\mathcal{I}=(a,b,N),\) an MMS allocation can be computed in polynomial time.
\end{theorem}

First, we provide a brief overview of the allocation procedure, and also define some important notions used in the procedure. Given two-chore-type instance \(\mathcal{I},\) we call the triangle computed by Lemma~\ref{lemma:alg-time} as the ``MMS triangle" for agent \(i,\) denoted by \(\triangle_i(\mathcal{I}).\) We denote by \(\mu_i(\mathcal{I})\) the corresponding multiplicities of its vertices.

\paragraph{Overview of allocation procedure.} The procedure works recursively. At each step, we assign bundles to a subset of agents and remove the corresponding chores and agents. The main invariant is that the feasibility point of the remaining instance remains in the intersection of the integer hulls of the remaining agents' feasible regions. This ensures that the MMS values of each remaining agent do not increase (due to Lemma~\ref{lemma:IDP}), and we can recurse on a smaller instance by finding an MMS allocation for the remaining instance.

Let us describe a simple case (corresponding to Case 1 in the proof) as an illustrative example. Suppose we find some integer point (bundle) \(X\) which lies ``top-right" of \(P_\mathcal{I},\) i.e. \(x(X)\geq x(P_\mathcal{I})\) and \(y(X)\geq y(P_\mathcal{I})\) (see Fig~\ref{fig:case-1-illustrative}).

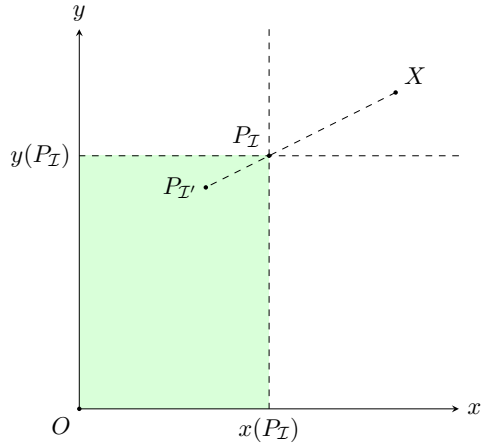
\begin{figure}[htbp]
    \centering
    \resizebox{0.4\linewidth}{!}{
    \begin{tikzpicture}[scale=1, >=stealth]
    
    \draw[->] (0,0) -- (6,0) node[right] {$x$};
    \draw[->] (0,0) -- (0,6) node[above] {$y$};
    
    \coordinate (P) at (3,4);
    \fill (P) circle (1pt) node[above left]{\(P_\mathcal{I}\)};
    
    \draw[dashed] (6,4) -- (0,4) node[left]{\(y(P_\mathcal{I})\)};
    \draw[dashed] (3,6) -- (3,0) node[below]{\(x(P_\mathcal{I})\)};
    
    \coordinate (X) at (5,5);
    \coordinate (P') at (2,3.5);
    \fill (X) circle (1pt) node[above right]{\(X\)};
    
    \draw[dashed] (X) -- (P');
    \fill (P') circle (1pt) node[left]{\(P_{\mathcal{I}'}\)};
    
    \coordinate (O) at (0,0);
    \fill (O) circle (1pt) node[below left]{\(O\)};
    
    \begin{scope}[on background layer]
        \fill[green!15] (O) -- (0,4) -- (P) -- (3,0) -- cycle;
    \end{scope}
    
    \end{tikzpicture}
    }
    \caption{(Case 1). Upon assigning bundle \(X\), the feasibility point moves from \(P_\mathcal{I}\) to \(P_{\mathcal{I}'}.\)} \label{fig:case-1-illustrative}
\end{figure}

Here \(P_\mathcal{I}=(a/n,b/n)\) represents a (possibly fractional) bundle with a proportional number of chores of each type. The disutility for this bundle will be exactly \(d_i(M)/n\) for any agent \(i.\) \(X\) represents a bundle which takes a ``more than proportional" amount of each type of chore (\(\geq a/n\) Type A chores, and \(\geq b/n\) Type B chores), so intuitively it should be an undesirable bundle (so getting rid of that set of chores would be ideal). If we can find some agent \(i\) which still finds bundle \(X\) \(\MMS\)-feasible, we can assign \(X\) to agent \(i.\) Consider the remaining instance \(\mathcal{I}'=(a-x(X),b-y(X),N\setminus\{i\}).\) Geometrically, the new feasibility point \(P_{\mathcal{I}'}\) must lie along the direction of \(\overrightarrow{XP_{\mathcal{I}}},\) and ``bottom-left" of \(P_\mathcal{I}.\) Since \(P_\mathcal{I}\) lay within \(\cap_{i\in N}IH_i,\) by convexity \(P_{\mathcal{I}'}\) must also lie within \(\cap_{i\in N'}IH_i.\) This is now a smaller instance and we can recurse.

This ``easy" case illustrates the general strategy. However, an integer bundle \(X\) ``top-right" of \(P_\mathcal{I}\) may not always exist (this corresponds to Cases 2a-e of the proof) and finding good points we can assign (and which agents to assign them to) is non-trivial, and requires performing case analysis by focusing on some well-chosen agent according to a particular partial order. We describe these cases formally in the proof.

\begin{proof}[Proof of Theorem~\ref{thm:main}]

Throughout the procedure, we say that a point \(X\) is ``top-right" of \(P_\mathcal{I}\) if \(x(X)\geq x(P_\mathcal{I})\) and \(y(X)\geq y(P_\mathcal{I}).\) The terms ``top-left", ``bottom-left", ``bottom-right" are also defined accordingly.

If \(n=1\), then there is a single agent \(i\in N\) with \(P_\mathcal{I}\in IH_i.\) Then we may simply assign the bundle \(P\) to agent \(i\) since it is \(\MMS\)-feasible for them (this is just assigning the entire set of chores to agent \(i\)). Assume the theorem holds for all \(n'<n\), and suppose \(n>1.\) 

First, we define a useful partial order.

\paragraph{Partial ordering of agents.} Let \(v_i\) be the value of \(y\) where \(\alpha_i\cdot x(P_\mathcal{I})+\beta_i\cdot y=\MMS_i,\) and let \(h_i\) be the value of \(x\) where \(\alpha_i\cdot x+\beta_i\cdot y(P_\mathcal{I})=\MMS_i.\) Define a partial order \(\preceq\) on the agents \(N\) so that for \(i,j\in N,\)
\[
i\preceq j \iff v_i\leq v_j \text{ and } h_i\leq h_j.
\]

Fix \(i\in N\) to be a minimal agent under \(\preceq\) (see Fig~\ref{fig:partial-order}).  Intuitively, \(L_i\) will be ``one of the closest" to \(P_\mathcal{I}\) among all the lines \((L_j)_{j\in N}.\) 

\begin{figure}[htbp]
    \centering
    \resizebox{0.4\linewidth}{!}{
    \begin{tikzpicture}[scale=1, >=stealth]
    \draw[->] (0,0) -- (6.5,0) node[right] {$x$};
    \draw[->] (0,0) -- (0,6.5) node[above] {$y$};
    
    \coordinate (P) at (1,1);
    \fill (P) circle (1pt) node[below left]{\(P_\mathcal{I}\)};
    
    \coordinate (I_j) at (0.8*6/2.8, -0.8*6/2.8 + 4);
    
    \draw[dashed] (1,0) node[below]{\(x(P)\)} -- (1,6);
    \draw[dashed] (0,1) node[left]{\(y(P)\)} -- (6,1);
    
    \draw[green!70!black, thick] (0,4) -- (4,0) node[pos=0.1, right, green!70!black] {\(L_i\)}; 
    
    \draw[dashed, green!70!black] (0,3) -- (1,3) node[pos=0, left, green!70!black] {\(v_i\)}; 
    
    \draw[dashed, green!70!black] (3,1) -- (3,0) node[pos=1, below, green!70!black] {\(h_i\)}; 
    
    \draw[red, thick] (0,5.5) -- (4.5,0) node[pos=0.1, right, red] {\(L_j\)}; 
    
    \draw[dashed, red] (0,-5.5/4.5 + 5.5) -- (1,-5.5/4.5 + 5.5) node[pos=0, left, red] {\(v_j\)}; 
    
    \draw[dashed, red] (4.5*4.5/5.5, 1) -- (4.5*4.5/5.5, 0) node[pos=1, below, red] {\(h_j\)}; 
    
    \draw[blue, thick] (0,3.2) -- (6,0) node[pos=0.8, right, blue] {\(L_k\)}; 
    
    \draw[dashed, blue] (0, -3.2/6+3.2) -- (1, -3.2/6+3.2) node[pos=0, left, blue] {\(v_k\)}; 
    
    \draw[dashed, blue] (2.2*6/3.2, 1) -- (2.2*6/3.2, 0) node[pos=1, below, blue] {\(h_k\)}; 
    
    \end{tikzpicture}
    }
    \caption{\(L_i\preceq L_j.\) \((L_i,L_j)\) and \((L_j,L_k)\) are incomparable. If \(x < x(P),\) then \(L_i\) is under \(L_j.\) If \(y<y(P),\) then \(L_i\) is under \(L_k.\)}\label{fig:partial-order}
    \end{figure}
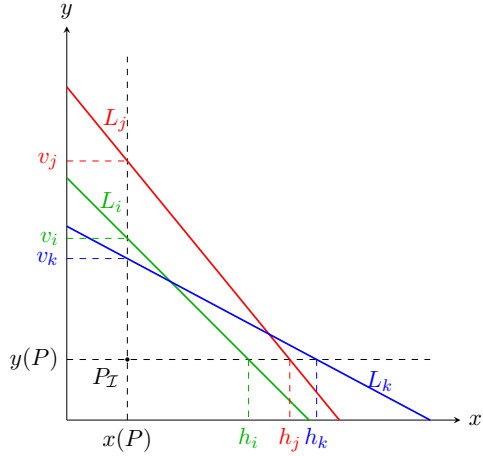

Depending on the geometric configuration of vertices of \(\triangle_i(\mathcal{I})\) (i.e. whether they lie ``top-left"/``bottom-left"/etc. of \(P_\mathcal{I}\)) and their corresponding multiplicities. We assign these bundles to different agents case-by-case.

By Observation~\ref{obs:P}, note that we know the feasibility point \(P_\mathcal{I}\in \cap_{i=1}^nIH_i\) for all \(i\in N.\)

Let \(i\in N\) be a minimal agent under \(\preceq.\) Now let \(A,B,C\) be the vertices of agent \(i\)'S MMS triangle \(\triangle_i(\mathcal{I}),\) and \(n_A,n_B,n_C\in \mu_i(\mathcal{I})\) their corresponding multiplicities. The procedure works recursively by allocating bundles \(A,B\) or \(C\) (possibly more than 1 of each) to satisfy some subset of agents, while ensuring that \(\MMS_i(\mathcal{I}')\leq \MMS_i(\mathcal{I})\) for all \(i\in N'\), where \(\mathcal{I}'\) is the remaining instance and \(N'\) is the remaining set of agents. We can then recurse by finding an MMS allocation for \(\mathcal{I}'.\)

In the allocation step, it suffices to show that after removing the corresponding bundles and agents, the new feasibility point, \(P_{\mathcal{I}'}\) remains in \(\cap_{i\in N'}IH_i.\) This is because Lemma~\ref{lemma:IDP} would imply that \(\MMS_i(\mathcal{I}')\leq \MMS_i(\mathcal{I}).\)

There are 5 cases in total (Case 1, and four subcases for Case 2). We prove the correctness of each case separately. \\

\noindent\textbf{Case 1:} Some \(X\in \{A,B,C\}\) lies ``top-right" of \(P_\mathcal{I}.\)

\emph{Allocation.} Assign \(X\) to agent \(i.\) The remaining instance is \(\mathcal{I}'=(a-x(X),b-y(X),N\setminus \{i\}).\)

\emph{Correctness.} See Fig~\ref{fig:case-1-illustrative} from the overview. Since \(X\in R_i\cap \mathbb{Z}^2,\) it is \(\MMS\)-feasible for agent \(i\) so we may assign them bundle \(X.\) For the remaining instance \(\mathcal{I}',\) the new feasibility point \(P_{\mathcal{I}'}=(\frac{a-x(X)}{n-1},\frac{b-y(X)}{n-1})\) will lie away from \(P_\mathcal{I}\) in the direction of \(\overrightarrow{XP_\mathcal{I}},\) which is ``bottom-left" of \(P_\mathcal{I}.\) Since \(a-x(X)\geq 0\) and \(b-y(X)\geq 0,\) \(P_{\mathcal{I}'}\) lies within the rectangle defined by the points \((0,0),(x(P_\mathcal{I}),0),P,(0,y(P_\mathcal{I}))\in R_j\cap \mathbb{Z}^2\) for all \(j\in N\). Thus by convexity \(P_{\mathcal{I}'}\in \cap_{j\in N\setminus \{i\}} IH_j\) and we can induct. \\

    
    
    
    
    
    
    
    

\noindent\textbf{Case 2:} None of the vertices \(A,B,C\) of \(\triangle_i(\mathcal{I})\) lie ``top-right" of \(P_\mathcal{I}.\) 

Now, observe that at least one of \(A,B,C\) must have \(y\)-coordinate which is greater than or equal to \(y(P_\mathcal{I})\) (since the weighted centroid of \(A,B,C\) is \(P_\mathcal{I}\)). WLOG, let this be \(B.\) If we are not in Case 1,  then we must have \(x(B)<x(P_\mathcal{I}),\) i.e. \(B\) is ``top-left" of \(P_\mathcal{I}.\) Similarly, we may assume that \(C\) lies ``bottom-right" of \(P_\mathcal{I}\) (i.e \(y(C)\leq y(P_\mathcal{I})\) and \(x(C)>x(P_\mathcal{I})\)), and WLOG that \(A\) lies under \(BC\) (so \(A,B,C\) are in clockwise order).

Next, define the set \(S:=\{j\in N\setminus \{i\}:\frac{\beta_j}{\alpha_j}\leq \frac{\beta_i}{\alpha_i}\}\) to be the set of agents \(j\in N\) such that the slope of the line \(L_j\) is steeper than that of \(L_i.\) Analogously, we define \(G:=\{j\in N\setminus \{i\}:\frac{\beta_j}{\alpha_j}> \frac{\beta_i}{\alpha_i}\}\) to be the set of agents \(j\in N\) such that the slope of the line \(L_j\) is gentler than that of \(L_i.\) 

Let \(I_{j}\) for \(j\in N\) denote the intersection point of lines \(L_i\) and \(L_j.\) For \(j\in N\), let \(R_j^\text{left}:=\{(x,y)\in R_j: x\leq x(P_\mathcal{I})\}\) and \(R_j^\text{below}:=\{(x,y)\in R_j: y\leq y(P_\mathcal{I})\}.\) We can make the following observation:

\begin{observation}\label{obs:steep-gentle}
    For all \(j\in S, R_i^\text{left}\subseteq R_j^\text{left},\) and for all \(k\in G, R_i^\text{below}\subseteq R_k^\text{below}\) (see  Fig~\ref{fig:partial-order-subset}).
\end{observation}

\begin{proof}
    Let \(j\in S.\) We show \(R_i^\text{left}\subseteq R_j^\text{left}\). Suppose \(x(I_j)<x(P_\mathcal{I}).\) Then, since \(L_j\) is steeper than \(L_i,\) we must have \(v_j<v_i\) and \(h_j<h_i,\) but this contradicts the minimality of \(i\) under \(\preceq.\) Thus \(x(I_j)\geq x(P_\mathcal{I}).\) Then suppose some point \(L\in R_i^{\text{left}}.\) Since \(L\) lies under \(L_i\), \(L_j\) is steeper than \(L_i\), and \(x(L)\leq x(P_\mathcal{I})\leq x(I_j),\) then \(L\in R_j^{\text{left}}.\) The argument for \(R_i^\text{below}\subseteq R_k^\text{below}\) for any \(k\in G\) is symmetric.
\end{proof}

Note that this may not necessarily imply that when \(x\leq x(P), IH_i\subseteq IH_j,\) or when \(y\leq y(P), IH_i\subseteq IH_k\).

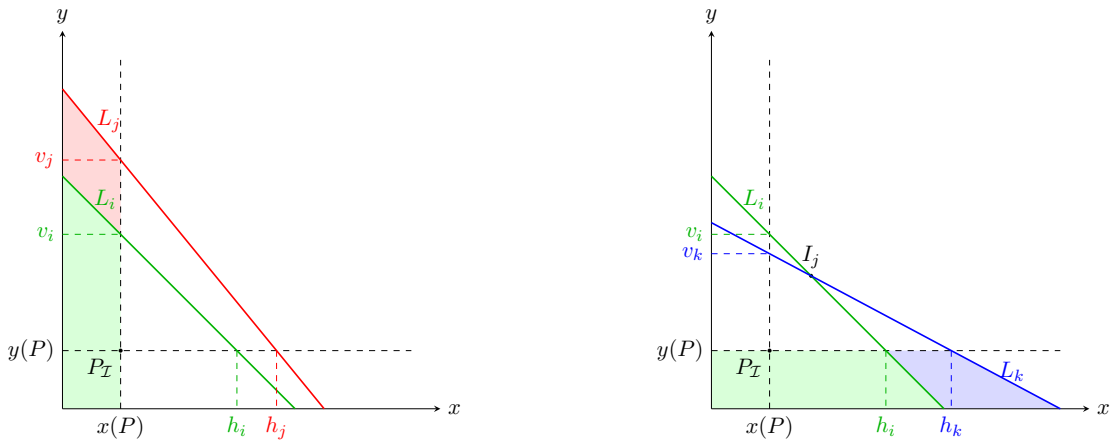
\begin{figure}[htbp]
    \centering
    \begin{subfigure}[t]{0.48\textwidth}
        \centering
        \resizebox{0.8\linewidth}{!}{
        \begin{tikzpicture}[scale=1, >=stealth]

        \draw[->] (0,0) -- (6.5,0) node[right] {$x$};
        \draw[->] (0,0) -- (0,6.5) node[above] {$y$};
        
        \coordinate (P) at (1,1);
        \fill (P) circle (1pt) node[below left]{\(P_\mathcal{I}\)};
        
        \coordinate (I_j) at (0.8*6/2.8, -0.8*6/2.8 + 4);
        
        \draw[dashed] (1,0) node[below]{\(x(P)\)} -- (1,6);
        \draw[dashed] (0,1) node[left]{\(y(P)\)} -- (6,1);
        
        \draw[green!70!black, thick] (0,4) -- (4,0) node[pos=0.1, right, green!70!black] {\(L_i\)}; 
        
        \draw[dashed, green!70!black] (0,3) -- (1,3) node[pos=0, left, green!70!black] {\(v_i\)}; 
        
        \draw[dashed, green!70!black] (3,1) -- (3,0) node[pos=1, below, green!70!black] {\(h_i\)}; 
        
        \draw[red, thick] (0,5.5) -- (4.5,0) node[pos=0.1, right, red] {\(L_j\)}; 
        
        \draw[dashed, red] (0,-5.5/4.5 + 5.5) -- (1,-5.5/4.5 + 5.5) node[pos=0, left, red] {\(v_j\)}; 
        
        \draw[dashed, red] (4.5*4.5/5.5, 1) -- (4.5*4.5/5.5, 0) node[pos=1, below, red] {\(h_j\)}; 
        
        \begin{scope}[on background layer]
            \fill[green!15] (O) -- (1,0) -- (1,3) -- (0,4) -- cycle;
        \end{scope}
        
        \begin{scope}[on background layer]
            \fill[red!15] (1,3) -- (0,4) -- (0,5.5) -- (1, -5.5/4.5 + 5.5) -- cycle;
        \end{scope}
        
        \end{tikzpicture}
}
    \end{subfigure}
    \hfill
    \begin{subfigure}[t]{0.48\textwidth}
        \centering
        \resizebox{0.8\linewidth}{!}{
        \begin{tikzpicture}[scale=1, >=stealth]
        
        \draw[->] (0,0) -- (6.5,0) node[right] {$x$};
        \draw[->] (0,0) -- (0,6.5) node[above] {$y$};
        
        \coordinate (P) at (1,1);
        \fill (P) circle (1pt) node[below left]{\(P_\mathcal{I}\)};
        
        \coordinate (I_j) at (0.8*6/2.8, -0.8*6/2.8 + 4);
        \fill (I_j) circle (1pt) node[above]{\(I_j\)};
        
        \draw[dashed] (1,0) node[below]{\(x(P)\)} -- (1,6);
        \draw[dashed] (0,1) node[left]{\(y(P)\)} -- (6,1);
        
        \draw[green!70!black, thick] (0,4) -- (4,0) node[pos=0.1, right, green!70!black] {\(L_i\)}; 
        
        \draw[dashed, green!70!black] (0,3) -- (1,3) node[pos=0, left, green!70!black] {\(v_i\)}; 
        
        \draw[dashed, green!70!black] (3,1) -- (3,0) node[pos=1, below, green!70!black] {\(h_i\)}; 
        
        \draw[blue, thick] (0,3.2) -- (6,0) node[pos=0.8, right, blue] {\(L_k\)}; 
        
        \draw[dashed, blue] (0, -3.2/6+3.2) -- (1, -3.2/6+3.2) node[pos=0, left, blue] {\(v_k\)}; 
        
        \draw[dashed, blue] (2.2*6/3.2, 1) -- (2.2*6/3.2, 0) node[pos=1, below, blue] {\(h_k\)}; 
        
        \begin{scope}[on background layer]
            \fill[green!15] (O) -- (0,1) -- (3,1) -- (4,0) -- cycle;
        \end{scope}
        
        \begin{scope}[on background layer]
            \fill[blue!15] (3,1) -- (4,0) -- (6,0) -- (2.2*6/3.2, 1) -- cycle;
        \end{scope}
        
        \end{tikzpicture}
        }
    \end{subfigure}
    \caption{Illustration of Obs~\ref{obs:steep-gentle}. \(R_i^\text{left}\subseteq R_j^\text{left}\) for \(j\in S\) (left) and \(R_i^\text{below}\subseteq R_k^\text{below}\) for \(k\in G\).}
    \label{fig:partial-order-subset}
\end{figure}

Now, we use the above fact to handle the remaining cases. \\

\noindent\textbf{Case 2a:} \(n_B\geq |S|\) \label{case:2a}

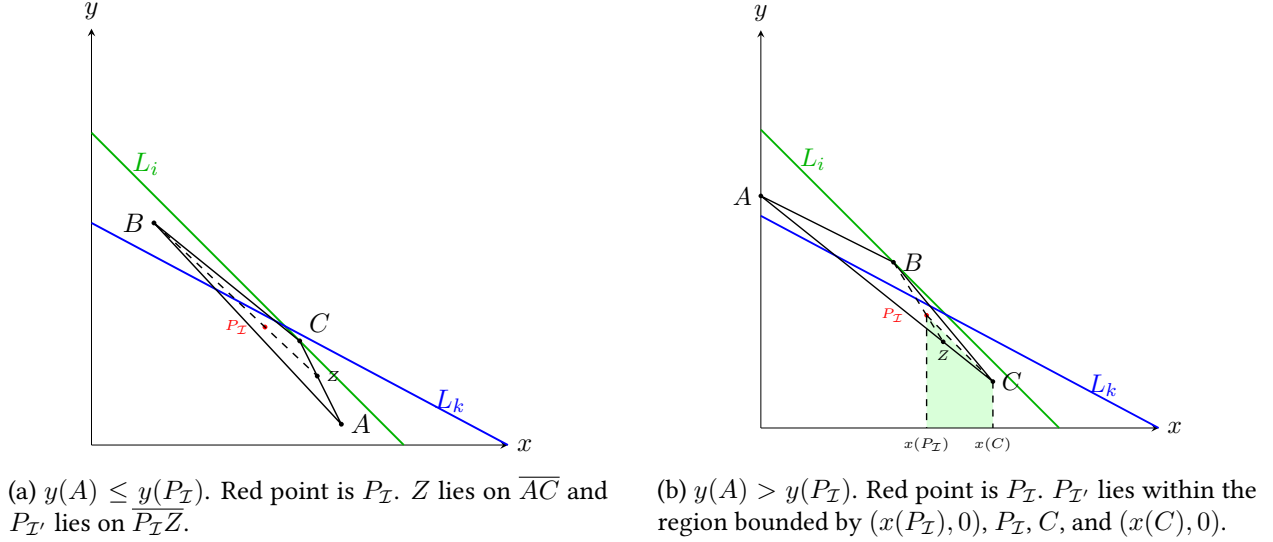
\begin{figure}[htbp]
    \centering
    \begin{subfigure}[t]{0.48\textwidth}
        \centering
        \resizebox{0.8\linewidth}{!}{
        \begin{tikzpicture}[scale=1, >=stealth]

        \draw[->] (0,0) -- (6,0) node[right] {$x$};
        \draw[->] (0,0) -- (0,6) node[above] {$y$};

        \coordinate (P) at (2.5,1.7);
        \fill[red] (P) circle (1pt)
            node[left, xshift=-3pt] {\tiny\(P_{\mathcal{I}}\)};

        \draw[green!70!black, thick] (0,4.5) -- (4.5,0)
            node[pos=0.1, right, green!70!black] {\(L_i\)};

        \draw[blue, thick] (0,3.2) -- (6,0)
            node[pos=0.8, right, blue] {\(L_k\)};

        \coordinate (A) at (3.6,0.3);
        \coordinate (B) at (0.9,3.2);
        \coordinate (C) at (3,1.5);

        \draw[line width=0.6pt, rounded corners=1pt]
            (A) -- (B) -- (C) -- cycle;
        \fill (A) circle (1pt) node[right] {\(A\)};
        \fill (B) circle (1pt) node[left] {\(B\)};
        \fill (C) circle (1pt) node[above right] {\(C\)};

        \coordinate (P') at (553/170,169/170);
        \fill (P') circle (1pt) node[right] {\tiny \(Z\)};
        \draw[dashed, line width=0.6pt] (B) -- (P');

        \end{tikzpicture}
        }
        \caption{\(y(A)\leq y(P_\mathcal{I})\). Red point is
        \(P_\mathcal{I}\). \(Z\) lies on \(\overline{AC}\) and
        \(P_{\mathcal{I}'}\) lies on \(\overline{P_\mathcal{I}Z}\).}
        \label{fig:case-2a-a}
    \end{subfigure}
    \hfill
    \begin{subfigure}[t]{0.48\textwidth}
        \centering
        \resizebox{0.8\linewidth}{!}{
        \begin{tikzpicture}[scale=1, >=stealth]

        \draw[->] (0,0) -- (6,0) node[right] {$x$};
        \draw[->] (0,0) -- (0,6) node[above] {$y$};

        \coordinate (P) at (2.5,1.7);
        \fill[red] (P) circle (1pt)
            node[left, xshift=-6pt] {\tiny\(P_{\mathcal{I}}\)};

        \draw[green!70!black, thick] (0,4.5) -- (4.5,0)
            node[pos=0.1, right, green!70!black] {\(L_i\)};

        \draw[blue, thick] (0,3.2) -- (6,0)
            node[pos=0.8, right, blue] {\(L_k\)};

        \coordinate (A) at (0,3.5);
        \coordinate (B) at (2,2.5);
        \coordinate (C) at (3.5,0.7);

        \draw[line width=0.6pt, rounded corners=1pt]
            (A) -- (B) -- (C) -- cycle;
        \fill (A) circle (1pt) node[left] {\(A\)};
        \fill (B) circle (1pt) node[right] {\(B\)};
        \fill (C) circle (1pt) node[right] {\(C\)};

        \coordinate (P') at (2.75,1.3);
        \fill (P') circle (1pt) node[below] {\tiny \(Z\)};
        \draw[dashed, line width=0.6pt] (B) -- (P');

        \draw[dashed, line width=0.6pt] (P) -- (2.5,0)
            node[below] {\tiny\(x(P_\mathcal{I})\)};
        \draw[dashed, line width=0.6pt] (P) -- (C);
        \draw[dashed, line width=0.6pt] (C) -- (3.5,0)
            node[below] {\tiny\(x(C)\)};

        \begin{scope}[on background layer]
            \fill[green!15] (P) -- (2.5,0) -- (3.5,0) -- (C) -- cycle;
        \end{scope}

        \end{tikzpicture}
        }
        \caption{\(y(A)>y(P_\mathcal{I})\). Red point is
        \(P_\mathcal{I}\). \(P_{\mathcal{I}'}\) lies within the
        region bounded by \((x(P_\mathcal{I}),0)\), \(P_\mathcal{I}\),
        \(C\), and \((x(C),0)\).}
        \label{fig:case-2a-b}
    \end{subfigure}

    \caption{Illustration of Case 2a.}
    \label{fig:case-2a}
\end{figure}

\emph{Allocation.} Assign each agent in \(S\) the bundle \(B.\) The remaining instance is \(\mathcal{I}'=(a-|S|\cdot x(B),b-|S|\cdot y(B),N\setminus S).\) 

\emph{Correctness.} See Fig~\ref{fig:case-2a}. By Observation~\ref{obs:steep-gentle}, the bundles are \(\MMS\)-feasible for all \(k\in S.\) Let the remaining instance be \(\mathcal{I}'=(a-|S|\cdot x(B),b-|S|\cdot y(B),N\setminus S)\) and note that \(N\setminus S = G\cup \{i\}.\) The new feasibility point \(P_{\mathcal{I}'}\) must lie along the segment \(\overline{P_\mathcal{I} Z}\) where \(Z = BP_\mathcal{I} \cap AC\) and ``bottom-right" of \(P_\mathcal{I}\). Now since \(P_{\mathcal{I}'}\) is still a weighted centroid of \(A,B,C\in IH_i,\) then \(P_{\mathcal{I}'}\in IH_i.\) 

\noindent Fix \(k\in G.\) By assumption, since \(y(C)\leq y(P_\mathcal{I}),\) by Observation~\ref{obs:steep-gentle}, we have \(C\in R_i^\text{below}\subseteq R_k^\text{below}\subseteq R_k.\) Since \(C\in\mathbb{Z}^2,\) then \(C\in IH_k.\) 

\noindent Suppose \(y(A) \leq y(P_\mathcal{I})\). Then by the same argument \(A\in IH_k.\) Hence since \(Z\) lies along \(\overline{AC}\), by convexity, \(Z\in IH_k\). 

\noindent Otherwise, suppose \(y(A) > y(P_\mathcal{I})\). Then observe that \(Z\) lies within the polygon defined by the points \((x(P_\mathcal{I}),0),P_\mathcal{I},C,(x(C),0).\) Since \(C\in IH_k,\) then \((x(C),0)\in IH_k.\) Also, by assumption, \(P_\mathcal{I}\in IH_k\) and hence \((x(P_\mathcal{I}),0)\in IH_k.\) Thus by convexity \(Z\in IH_k.\)

\noindent In either case, we have \(Z\in \cap_{k\in N\setminus S}IH_k\), hence \(P_{\mathcal{I}'}\in \cap_{k\in N\setminus S}IH_k\) by the fact that \(Z, P_{\mathcal{I}}\in \cap_{k\in N\setminus S}IH_k\) and the convexity of \(\cap_{k\in N\setminus S}IH_k\), and we may induct. \\

\noindent\textbf{Case 2b:} \(n_C\geq |G|.\)

\emph{Allocation.} Assign each agent in \(G\) the bundle \(C.\) The remaining instance is \(\mathcal{I}'=(a-|G|\cdot x(C),b-|G|\cdot y(C),N\setminus G).\) 

\emph{Correctness.} We handles this symmetrically to Case 2a by assigning each agent in \(G\) the bundle \(C\); the proof follows by exchanging the axes. \\

\noindent\textbf{Case 2c:} \(n_B<|S|, n_C<|G|\) and \(A\) lies ``bottom-left" of \(P_\mathcal{I}\) 

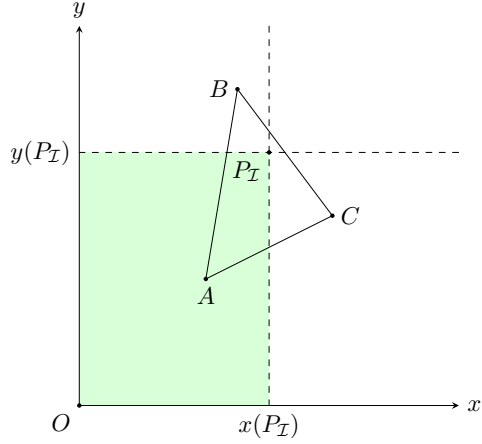
\begin{figure}[htbp]
    \centering
    \resizebox{0.4\linewidth}{!}{
    \begin{tikzpicture}[scale=1, >=stealth]
    
    \draw[->] (0,0) -- (6,0) node[right] {$x$};
    \draw[->] (0,0) -- (0,6) node[above] {$y$};
    
    \coordinate (P) at (3,4);
    \fill (P) circle (1pt) node[below left]{\(P_\mathcal{I}\)};
    
    \draw[dashed] (6,4) -- (0,4) node[left]{\(y(P_\mathcal{I})\)};
    \draw[dashed] (3,6) -- (3,0) node[below]{\(x(P_\mathcal{I})\)};
    
    \coordinate (O) at (0,0);
    \fill (O) circle (1pt) node[below left]{\(O\)};
    
    \begin{scope}[on background layer]
        \fill[green!15] (O) -- (0,4) -- (P) -- (3,0) -- cycle;
    \end{scope}
    
    \coordinate (B) at (2.5,5);
    \fill (B) circle (1pt) node[left] {\(B\)};
    
    \coordinate (C) at (4,3);
    \fill (C) circle (1pt) node[right] {\(C\)};
    
    \coordinate (A) at (2,2);
    \fill (A) circle (1pt) node[below] {\(A\)};
    
    \draw (A) -- (B) -- (C) -- cycle;
    
    \end{tikzpicture}
    }
    \caption{(Case 2c) \(A\) lies ``bottom-left" of \(P_\mathcal{I}\) hence is \(D_i\)-feasible for all \(i\in N.\)} \label{fig:case-2c}
\end{figure}

\emph{Allocation.} Assign all \(n_B\) bundles \(B\) to agents in \(S,\) and all \(n_C\) bundles \(C\) to agents in \(G.\)

\emph{Correctness.} See Fig~\ref{fig:case-2c}. By Observation~\ref{obs:steep-gentle}, these are MMS-feasible assignments. 

Now since \(A\) is ``bottom-left" of \(P_\mathcal{I},\) and \(P_\mathcal{I}\) lies under \(L_j\) for all \(j\in N,\) i.e. \(A\) is \(\MMS\)-feasible for all agents \(j\in N\). Hence we may simply assign \(A\) to any remaining agent, and since \(n_A\cdot A + n_B\cdot B + n_C\cdot C=(a,b)\) (by definition), this is a complete allocation and an \(\MMS\)-feasible one. \\

\noindent\textbf{Case 2d:} \(n_B<|S|, n_C<|G|\) and \(A\) lies ``top-left" of \(P_\mathcal{I}.\)

\begin{figure}[htbp]
    \centering
    \resizebox{0.4\linewidth}{!}{
    \begin{tikzpicture}[scale=1, >=stealth]

    \draw[->] (0,0) -- (6,0) node[right] {$x$};
    \draw[->] (0,0) -- (0,6) node[above] {$y$};
    
    \coordinate (P) at (2.5,1.7);
    \fill[red] (P) circle (1pt);
    
    \draw[green!70!black, thick] (0,4.5) -- (4.5,0) node[pos=0.1, right, green!70!black] {\(L_i\)}; 
    
    \draw[blue, thick] (0,3.2) -- (6,0) node[pos=0.8, right, blue] {\(L_k\)}; 
    
    \coordinate (A) at (0,3.5);
    \coordinate (B) at (2,2.5);
    \coordinate (C) at (3.8,0.6);
    
    \draw[line width=0.6pt, rounded corners=1pt] (A) -- (B) -- (C) -- cycle;
    \fill (A) circle (1pt) node[left] {\(A\)};
    \fill (B) circle (1pt) node[right] {\(B\)};
    \fill (C) circle (1pt) node[right] {\(C\)};
    
    \coordinate (Q) at (418/159,475/318);
    \fill (Q) circle (1pt) node[below] {\tiny\(Q\)};
    \draw[dashed, line width=0.6pt] (B) -- (Q);
    
    \coordinate (P') at (3.2,1.05);
    \fill (P') circle (1pt) node[below] {\tiny \(P_{\mathcal{I}'}\)};
    
    \draw[dashed, line width=0.6pt] (P) -- (2.5,0) node[below] {\tiny\(x(P_\mathcal{I})\)};
    \draw[dashed, line width=0.6pt] (P) -- (C);
    \draw[dashed, line width=0.6pt] (C) -- (3.8,0) node[below] {\tiny\(x(C)\)};
    
    \begin{scope}[on background layer]
        \fill[green!15] (P) -- (C) -- (3.8,0) -- (2.5,0) -- cycle;
    \end{scope}
    
    \end{tikzpicture}
    }
    \caption{(Case 2d)} \label{fig:case-2d}
\end{figure}

\emph{Allocation.} Assign \(n_B\) agents in \(S\) the bundle \(B.\) Assign \(|S|-n_B\) remaining agents in \(S\) the bundle \(A.\) The remaining instance is \(\mathcal{I}'=(a-n_B\cdot x(B)-(|S|-n_B)\cdot x(A),b-n_B\cdot y(B)-(|S|-n_B)\cdot y(A),N\setminus S)\)

\emph{Correctness.} See Fig~\ref{fig:case-2d}. Note that there are enough bundles since \(n_A+n_B+n_C=|S|+|G|+1\) and \(n_C<|G|\) implies \(n_A+n_B>|S|.\) By Observation~\ref{obs:steep-gentle}, these are \(D_j\)-feasible for all \(j\in S.\) Let the remaining instance be \(\mathcal{I}'=(a-n_B\cdot x(B)-(|S|-n_B)\cdot x(A),b-n_B\cdot y(B)-(|S|-n_B)\cdot y(A),N\setminus S)\) and note that \(N\setminus S=G\cup \{i\}.\) Let \(Q:=\frac{n_A\cdot A+n_C\cdot C}{n_A+n_C}\) be the weighted centroid of the remaining points after removing bundles \(B.\) This must lie along the segment \(\overline{AC},\) and is ``bottom-right" of \(P_\mathcal{I}\). Now upon removing \(|S|-n_B\) copies of \(A,\) the new feasibility point \(P_{\mathcal{I}'}\) must lie along the segment \(\overline{QC},\) and is still ``bottom-right" of \(P_\mathcal{I}.\) Thus, like in Case 2a, this is within the polygon defined by the points \((x(P_\mathcal{I}),0),P_\mathcal{I},C,(x(C),0).\) Then by the same argument \(P_{\mathcal{I}'}\in \cap_{j\in N\setminus S} IH_j(D_j)\) and we may induct. \\

\noindent\textbf{Case 2e:} \(n_B<|S|, n_C<|G|\) and \(A\) lies ``bottom-right" of \(P_\mathcal{I}.\)

\emph{Allocation.} Assign \(n_C\) agents in \(G\) the bundle \(C.\) Assign \(|G|-n_C\) remaining agents in \(G\) the bundle \(A.\) The remaining instance is \(\mathcal{I}'=(a-n_C\cdot x(C)-(|G|-n_C)\cdot x(A),b-n_C\cdot y(C)-(|G|-n_C)\cdot y(A),N\setminus G)\)

\emph{Correctness.} We handle this symmetrically to Case 2d by assigning all \(n_C\) bundles \(C\) to agents in \(G,\) and \(|G|-n_C\) copies of \(A\) to the remaining agents in \(G.\) The proof follows by exchanging the axes. 

\end{proof}

\textbf{Time complexity.} By Lemma~\ref{lemma:mms-identical}, the values \((\MMS_j)_{j\in N}\) may be computed in polynomial time. In each iteration, we may find the minimal agent \(i\in N\) under \(\preceq\) in \(O(n)\) time. Computing \(\triangle_i(\mathcal{I})\) also takes \(O(\log^2 \MMS_i)\) time by Lemma~\ref{lemma:alg-time}, and the number of agents to consider after each iteration decreases by at least one. Thus the allocation procedure can be done in polynomial time.


\section{Extension to goods and mixed items}\label{section:Extension}

In this section, we briefly note that the procedure in Section~\ref{section:main-result} can also be modified to show the existence of MMS allocations when the two types of items are both goods, or when there is one type of item which is an \textit{objective good} (positive utility to all agents), and another type of item which is an \textit{objective chore} (negative utility to all agents). Moreover, this can also be done in polynomial time. The proofs make use of essentially the same techniques. 

Instead of working with disutilities, we use regular agent valuation functions \(v_i:M\to \mathbb{R}.\) An item \(o\) is a good (resp. chore) for agent \(i\) if \(v_i(o)>0\) (\(v_i(o)<0\)). Then recall in this setting the maximin share of an agent is defined by
\[
\MMS_i(M) := \max_{(X_i)_{i\in N}\in\Pi}\min_{j\in N}v_i(X_j),
\]
so an MMS allocation is one where each agent receives \textit{at least} their MMS value.
We re-define MMS-feasibility to represent if the value of a bundle \(B\) for some agent \textit{exceeds} their MMS value.

Observe that, in an fair-division instance with goods, the set of \(\MMS_i\)-feasible bundles for agent \(i\), \(R_i\) now refers to the nonnegative region lying \textit{above} the line \(\alpha_ix+\beta_iy=\MMS_i\) (where \(\alpha_i,\beta_i>0\)). Similarly, when there is one type of \textit{objective good} and one type of \textit{objective chore}, WLOG let Type A refer to the goods and Type B refer to the chores. Then \(R_i\) refers to the nonnegative region lying \textit{below} the line \(\alpha_ix-\beta_iy=\MMS_i\) where \(\alpha_i,\beta_i>0.\) One can show, using essentially the same proof, that in these settings Lemma~\ref{lemma:IDP} still holds.

\begin{figure}[htbp]
    \centering
    \begin{subfigure}[t]{0.48\textwidth}
        \centering
        \resizebox{0.8\linewidth}{!}{
        \begin{tikzpicture}[scale=2.5, >=stealth]

        \draw[->] (0,0) -- (2.2,0) node[right] {$x$};
        \draw[->] (0,0) -- (0,2.2) node[above] {$y$};
        
        \draw[red, thick] (0,2) -- (1.2,0) node[pos=0.91, left, red] {\small $L_i$};
        
        \def\x{0.375}   
        \def\y{0.375}     
        
        \foreach \i in {0,...,5} {
            \foreach \j in {0,...,5} {
            \fill[black] ({\i*\x}, {\j*\y}) circle (0.6pt);
            }
        }
        
        \begin{scope}[on background layer]
            \fill[red!20] (\x*4, 0) -- (\x*2, \y*2) -- (\x*1, \y*4) -- (\x*0.3, \y*5.8) -- (\x*5.5, \y*5.8) -- (\x*5.5, 0) -- cycle;
        \end{scope}
        
        
        \coordinate (P) at (2*\x,2.25*\y);
        \fill[red] (P) circle (0.4pt);
        
        \draw (2*\x,2*\y) -- (1*\x,4*\x) -- (3*\x,1*\y) -- cycle;

        
        \end{tikzpicture}   
        }
        \caption{Goods instance. \(L_i\) is \(5x+3y=16.\) \( P_\mathcal{I}=(2,\frac{9}{4})=\frac{1}{4}((2\cdot (2,2)+(1,4)+(3,1)).\) Shaded region is \(IH_i.\)}
    \end{subfigure}
    \hfill
    \begin{subfigure}[t]{0.48\textwidth}
        \centering
        \resizebox{0.8\linewidth}{!}{
        \begin{tikzpicture}[scale=2.5, >=stealth]

        \draw[->] (0,0) -- (2.2,0) node[right] {$x$};
        \draw[->] (0,0) -- (0,2.2) node[above] {$y$};
        
        \def\x{0.375}   
        \def\y{0.375}     
        
        \draw[blue, thick] (0.2*\x,0) -- (3.68*\x,5.8*\y) node[pos=0.77, left, blue] {\small $L_i$};
        
        \foreach \i in {0,...,5} {
            \foreach \j in {0,...,5} {
            \fill[black] ({\i*\x}, {\j*\y}) circle (0.6pt);
            }
        }
        
        \begin{scope}[on background layer]
            \fill[blue!20] (2*\x,3*\y) -- (3.68*\x, 5.8*\y) -- (5.5*\x, 5.8*\y) -- (5.5*\x, 0) -- (1*\x, 0) -- (1*\x, 1*\y) -- cycle;
        \end{scope}
        
        
        \coordinate (P) at (2.25*\x,3*\y);
        \fill[red] (P) circle (0.4pt);
        
        \draw (2*\x,3*\y) -- (3*\x,4*\y) -- (2*\x,2*\y) -- cycle;

        
        \end{tikzpicture}   
        }
        \caption{Mixed instance. \(L_i\) is \(5x-3y=1.\) \(P_\mathcal{I}=(\frac{9}{4},3)=\frac{1}{4}(2\cdot (2,3)+(3,4)+(2,2)).\) Shaded region is \(IH_i.\)}
    \end{subfigure}
    \label{fig:extensions}
\end{figure}

Although the regions \(R_i\) may now be unbounded, it is still possible to find a unimodular triangle with vertices in \(R_i\cap \mathbb{Z}^2\) containing \(P_\mathcal{I}\) in polynomial time, i.e. Lemma~\ref{lemma:alg-time} can be extended (Theorem 6.1 in \cite{mccormick_polynomial_1997}).

Finally, we may also define \(\preceq\) analogously in each case. Recall that \(v_i\) is the value of \(y\) where \(\alpha_i\cdot x(P_\mathcal{I})+\beta_i\cdot y=\MMS_i,\) and \(h_i\) is the value of \(x\) where \(\alpha_ix+\beta_i\cdot y(P_\mathcal{I})=\MMS_i.\) For an instance of two types of goods, we say that
\[
i \preceq j \iff v_i\geq v_j \text{ and } h_i \geq h_j.
\]
For a mixed instance, we instead say that
\[
i \preceq j \iff v_i\leq v_j \text{ and } h_i \geq h_j.
\]

Following the above definitions, extending the allocation procedure to these cases essentially follow the same case analyses. Hence we have the following results.

\begin{restatable}{theorem}{goods}\label{thm:goods}
    Given an instance of fair allocation with only two types of goods, there exists an \(\MMS\) allocation which can be computed in polynomial time.
\end{restatable}

\begin{restatable}{theorem}{mixed}\label{thm:mixed}
    Given an instance of fair allocation where there are only two types of items; one of which being objective goods, and the other being objective chores, there exists an \(\MMS\) allocation which can be computed in polynomial time.
\end{restatable}

We provide a full proof of the mixed case in Appendix~\ref{sec:mixed}, however we omit a proof of the goods case as the geometric techniques are very similar.

\section{Discussion}

We show the existence (and polynomial-time computation) of MMS allocations for instances with two types of chores, two types of goods, or one type of objective good and one type of objective chore, by leveraging a novel geometric perspective inspired by prior work in job scheduling for two types of jobs.

However, our technique does not extend easily when the number of types of items is \(>2,\) as the feasibility point \(P_\mathcal{I}\) may not necessarily lie within some unimodular lattice simplex under the hyperplane \(L_i\) \citep{mccormick_polynomial_1997}. One way to see this is because Pick's theorem does not extend to dimensions \(d>2,\) so simplices with integer vertices containing no lattice point within them may no longer neccessarily be unimodular. One possible future direction could be to consider ``two-type" valuation curves more general than additive, e.g. \(v_i(S):=f(x,y),\) where \(x\) and \(y\) are the number of Type A and Type B chores in \(S\) respectively.

\bibliographystyle{plainnat}
\bibliography{references}

@inproceedings{aziz2017,
author = {Aziz, Haris and Rauchecker, Gerhard and Schryen, Guido and Walsh, Toby},
title = {Algorithms for max-min share fair allocation of indivisible chores},
year = {2017},
publisher = {AAAI Press},
booktitle = {Proceedings of the Thirty-First AAAI Conference on Artificial Intelligence},
pages = {335–341},
numpages = {7},
location = {San Francisco, California, USA},
series = {AAAI'17}
}

@inproceedings{conf/soda/HeidariKSS26,
author = {Ehsan Heidari and Alireza Kaviani and Masoud Seddighin and AmirMohammad Shahrezaei},
title = {Improved Maximin Share Guarantee for Additive Valuations},
booktitle = {Proceedings of the 2026 Annual ACM-SIAM Symposium on Discrete Algorithms (SODA)},
year = {2026},
chapter = {},
pages = {2239-2290},
doi = {10.1137/1.9781611978971.81},
URL = {https://epubs.siam.org/doi/abs/10.1137/1.9781611978971.81}
}

@inproceedings{conf/soda/AkramiG24,
author = {Hannaneh Akrami and Jugal Garg},
title = {Breaking the 3/4 Barrier for Approximate Maximin Share},
booktitle = {Proceedings of the 2024 Annual ACM-SIAM Symposium on Discrete Algorithms (SODA)},
chapter = {},
year = {2024},
pages = {74-91},
doi = {10.1137/1.9781611977912.4},
URL = {https://epubs.siam.org/doi/abs/10.1137/1.9781611977912.4}
}

@article{journals/mor/GhodsiHSSY21,
author = {Ghodsi, Mohammad and Hajiaghayi, Mohammad Taghi and Seddighin, Masoud and Seddighin, Saeed and Yami, Hadi},
title = {Fair Allocation of Indivisible Goods: Improvement},
journal = {Mathematics of Operations Research},
volume = {46},
number = {3},
pages = {1038-1053},
year = {2021},
doi = {10.1287/moor.2020.1096},
URL = {https://doi.org/10.1287/moor.2020.1096}
}

@inproceedings{ebadian2022,
author = {Ebadian, Soroush and Peters, Dominik and Shah, Nisarg},
title = {How to Fairly Allocate Easy and Difficult Chores},
year = {2022},
isbn = {9781450392136},
publisher = {International Foundation for Autonomous Agents and Multiagent Systems},
address = {Richland, SC},
booktitle = {Proceedings of the 21st International Conference on Autonomous Agents and Multiagent Systems},
pages = {372–380},
numpages = {9},
location = {Virtual Event, New Zealand},
series = {AAMAS '22}
}

@article{amanatidis2017,
author = {Amanatidis, Georgios and Markakis, Evangelos and Nikzad, Afshin and Saberi, Amin},
title = {Approximation Algorithms for Computing Maximin Share Allocations},
year = {2017},
issue_date = {October 2017},
publisher = {Association for Computing Machinery},
address = {New York, NY, USA},
volume = {13},
number = {4},
issn = {1549-6325},
url = {https://doi.org/10.1145/3147173},
doi = {10.1145/3147173},
journal = {ACM Trans. Algorithms},
month = dec,
articleno = {52},
numpages = {28}
}

@misc{zhang2026exactmmsallocationspersonalized,
      title={Exact MMS Allocations under Personalized Bivalued Valuations: Goods and Chores}, 
      author={Yuhao Zhang},
      year={2026},
      eprint={2608.15822},
      archivePrefix={arXiv},
      primaryClass={cs.GT},
      url={https://arxiv.org/abs/2608.15822}, 
}

@article{marcotte1985cutting,
  author  = {Marcotte Odile},
  title   = {The Cutting Stock Problem and Integer Rounding},
  journal = {Mathematical Programming},
  volume  = {33},
  pages   = {82--92},
  year    = {1985},
  doi     = {10.1007/BF01582013}
}

@article{budish2011,
 ISSN = {00223808, 1537534X},
 author = {Eric Budish},
 journal = {Journal of Political Economy},
 number = {6},
 pages = {1061--1103},
 publisher = {The University of Chicago Press},
 title = {The Combinatorial Assignment Problem: Approximate Competitive Equilibrium from Equal Incomes},
 urldate = {2026-05-09},
 volume = {119},
 year = {2011}
}

@misc{huang2025fptas79approximationmaximinshare,
      title={An FPTAS for 7/9-Approximation to Maximin Share Allocations}, 
      author={Xin Huang and Shengwei Zhou},
      year={2025},
      eprint={2511.13056},
      archivePrefix={arXiv},
      primaryClass={cs.GT},
      url={https://arxiv.org/abs/2511.13056}, 
}

@article{viswanathan_general_2025,
	title = {A {General} {Framework} for {Fair} {Allocation} under {Matroid} {Rank} {Valuations}},
	volume = {13},
	issn = {2167-8375},
	url = {https://dl.acm.org/doi/10.1145/3728367},
	doi = {10.1145/3728367},
	number = {3},
	urldate = {2026-05-08},
	journal = {ACM Trans. Econ. Comput.},
	author = {Viswanathan, Vignesh and Zick, Yair},
	month = sep,
	year = {2025},
	pages = {14:1--14:32},
}

@inproceedings{barman_fair_2023,
	series = {{AAMAS} '23},
	title = {Fair {Chore} {Division} under {Binary} {Supermodular} {Costs}},
	url = {https://dl.acm.org/doi/10.5555/3545946.3599104},
	urldate = {2026-05-08},
	booktitle = {Proceedings of the 2023 {International} {Conference} on {Autonomous} {Agents} and {Multiagent} {Systems}},
	author = {Barman, Siddharth and Narayan, Vishnu and Verma, Paritosh},
	month = may,
	year = {2023},
	pages = {2863--2865},
}

@inproceedings{barman_existence_2021,
	series = {{AAMAS} '21},
	title = {Existence and {Computation} of {Maximin} {Fair} {Allocations} {Under} {Matroid}-{Rank} {Valuations}},
	url = {https://dl.acm.org/doi/10.5555/3463952.3463978},
	urldate = {2026-05-08},
	booktitle = {Proceedings of the 20th {International} {Conference} on {Autonomous} {Agents} and {MultiAgent} {Systems}},
	author = {Barman, Siddharth and Verma, Paritosh},
	month = may,
	year = {2021},
	pages = {169--177},
}

@inproceedings{gorantla_fair_2023,
	title = {Fair allocation of a multiset of indivisible items},
	url = {https://epubs.siam.org/doi/10.1137/1.9781611977554.ch13},
	urldate = {2026-05-08},
	booktitle = {{ACM}-{SIAM} {Symposium} on {Discrete} {Algorithms} ({SODA})},
	author = {Gorantla, Pranay and Marwaha, Kunal and Velusamy, Santhoshini},
	month = jan,
	year = {2023},
	doi = {10.1137/1.9781611977554.ch13},
	pages = {304--331},
}

@article{aziz_fair_2021,
	title = {Fair allocation of indivisible goods and chores},
	volume = {36},
	issn = {1573-7454},
	url = {https://doi.org/10.1007/s10458-021-09532-8},
	doi = {10.1007/s10458-021-09532-8},
	language = {en},
	number = {1},
	urldate = {2026-05-08},
	journal = {Autonomous Agents and Multi-Agent Systems},
	author = {Aziz, Haris and Caragiannis, Ioannis and Igarashi, Ayumi and Walsh, Toby},
	month = nov,
	year = {2021},
	pages = {3},
}

@article{kurokawa_fair_2018,
	title = {Fair {Enough}: {Guaranteeing} {Approximate} {Maximin} {Shares}},
	volume = {65},
	issn = {0004-5411},
	shorttitle = {Fair {Enough}},
	url = {https://dl.acm.org/doi/10.1145/3140756},
	doi = {10.1145/3140756},
	number = {2},
	urldate = {2026-05-08},
	journal = {J. ACM},
	author = {Kurokawa, David and Procaccia, Ariel D. and Wang, Junxing},
	month = feb,
	year = {2018},
	pages = {8:1--8:27},
}

@inproceedings{aziz_fair_2023,
	series = {{AAMAS} '23},
	title = {Fair {Allocation} of {Two} {Types} of {Chores}},
	urldate = {2025-10-28},
	booktitle = {Proceedings of the 2023 {International} {Conference} on {Autonomous} {Agents} and {Multiagent} {Systems}},
	author = {Aziz, Haris and Lindsay, Jeremy and Ritossa, Angus and Suzuki, Mashbat},
	month = may,
	year = {2023},
	pages = {143--151},
}

@article{feige_maximin_2022,
	title = {Maximin fair allocations with two item values},
	url = {https://www.wisdom.weizmann.ac.il/~feige/mypapers/MMSab.pdf},
	urldate = {2025-10-28},
	journal = {Unpublished Manuscript},
	author = {Feige, Uriel},
	year = {2022},
}

@inproceedings{feige_tight_2021,
	title = {A {Tight} {Negative} {Example} for {MMS} {Fair} {Allocations}},
	url = {https://doi.org/10.1007/978-3-030-94676-0_20},
	doi = {10.1007/978-3-030-94676-0_20},
	urldate = {2025-10-28},
	booktitle = {Proceedings of the 17th International Conference on Web and {Internet} {Economics} {(WINE)}},
	author = {Feige, Uriel and Sapir, Ariel and Tauber, Laliv},
	month = dec,
	year = {2021},
    pages = {355--372}
}

@inproceedings{garg_improved_2025,
	title = {Improved maximin share approximations for chores by bin packing},
	url = {https://doi.org/10.1609/aaai.v39i13.33518},
	doi = {10.1609/aaai.v39i13.33518},
	urldate = {2025-10-28},
	booktitle = {Proceedings of the {Thirty}-{Ninth} {AAAI} {Conference} on {Artificial} {Intelligence}},
	author = {Garg, Jugal and Huang, Xin and Segal-Halevi, Erel},
	month = aug,
	year = {2025},
	pages = {13881--13888},
}

@inproceedings{huang_reduction_2023,
	series = {{EC} '23},
	title = {A {Reduction} from {Chores} {Allocation} to {Job} {Scheduling}},
	url = {https://doi.org/10.1145/3580507.3597676},
	doi = {10.1145/3580507.3597676},
	urldate = {2025-10-28},
	booktitle = {Proceedings of the 24th {ACM} {Conference} on {Economics} and {Computation}},
	author = {Huang, Xin and Segal-Halevi, Erel},
	month = jul,
	year = {2023},
	pages = {908},
}

@inproceedings{mccormick_polynomial_1997,
	series = {{SODA} '97},
	title = {Polynomial algorithms for multiprocessor scheduling with a small number of job lengths},
	urldate = {2025-10-26},
	booktitle = {Proceedings of the eighth annual {ACM}-{SIAM} symposium on {Discrete} algorithms},
	author = {McCormick, S. Thomas and Smallwood, Scott R. and Spieksma, Frits C. R.},
	month = jan,
	year = {1997},
	pages = {509--517},
}

@inproceedings{huang_algorithmic_2021,
	series = {{EC} '21},
	title = {An {Algorithmic} {Framework} for {Approximating} {Maximin} {Share} {Allocation} of {Chores}},
	url = {https://doi.org/10.1145/3465456.3467555},
	doi = {10.1145/3465456.3467555},
	urldate = {2025-10-21},
	booktitle = {Proceedings of the 22nd {ACM} {Conference} on {Economics} and {Computation}},
	author = {Huang, Xin and Lu, Pinyan},
	month = jul,
	year = {2021},
	pages = {630--631},
}

@inproceedings{aziz_developments_2020,
	title = {Developments in Multi-Agent Fair Allocation},
	url = {https://doi.org/10.1609/aaai.v34i09.7082},
	doi = {10.1609/aaai.v34i09.7082},
	urldate = {2025-10-28},
	booktitle = {Proceedings of the {Thirty-Fourth} {AAAI} {Conference} on {Artificial} {Intelligence}},
	author = {Aziz, Haris},
	month = apr,
	year = {2020},
	pages = {13563--13568},
}

@article{10.1016/j.artint.2023.103965,
author = {Amanatidis, Georgios and Aziz, Haris and Birmpas, Georgios and Filos-Ratsikas, Aris and Li, Bo and Moulin, Herv\'{e} and Voudouris, Alexandros A. and Wu, Xiaowei},
title = {Fair division of indivisible goods: Recent progress and open questions},
year = {2023},
issue_date = {Sep 2023},
publisher = {Elsevier Science Publishers Ltd.},
address = {GBR},
volume = {322},
number = {C},
issn = {0004-3702},
url = {https://doi.org/10.1016/j.artint.2023.103965},
doi = {10.1016/j.artint.2023.103965},
journal = {Artif. Intell.},
month = sep,
numpages = {25}
}

@article{liu2024mixed,
  author    = {Shengxin Liu and Xinhang Lu and Mashbat Suzuki and Toby Walsh},
  title     = {Mixed Fair Division: A Survey},
  journal   = {Proceedings of the AAAI Conference on Artificial Intelligence},
  volume    = {38},
  number    = {20},
  pages     = {22641--22649},
  year      = {2024},
  month     = mar,
  doi       = {10.1609/aaai.v38i20.30274},
  url       = {https://doi.org/10.1609/aaai.v38i20.30274}
}

@article{gates_helpful_2020,
  author  = {Gates, Victoria and Griffiths, Thomas L. and Dragan, Anca D.},
  title   = {How to Be Helpful to Multiple People at Once},
  journal = {Cognitive Science},
  volume  = {44},
  number  = {6},
  pages   = {e12841},
  year    = {2020},
  doi     = {10.1111/cogs.12841},
}

\appendix

\section{Appendix} \label{sec:appendix}

\subsection{Integer decomposition lemmas} \label{sec:appendix-IDP}

We prove Lemma~\ref{lemma:IDP} by including a proof of a more general statement, Lemma~\ref{lemma:IDP-general}; which was also established in \cite{mccormick_polynomial_1997}. We use the latter lemma for the proof of Theorem~\ref{thm:mixed}. These lemmas are related to what is known as the \textit{integer decomposition property} \cite{marcotte1985cutting}.

\begin{lemma}[Adapted from Corollary 3.6, Theorem 6.1 in \cite{mccormick_polynomial_1997}]\label{lemma:IDP-general}
Let \(P=(a/n,b/n)\) for integers \(a,b,n.\) Let \(R\subseteq \mathbb{R}^2\) be a (possibly unbounded) polyhedron, and let \(H\) be its integer hull. The following statements are equivalent.
\begin{enumerate}[label=(\roman*)]
    \item \(P\in H\)
    \item \(P\) lies in some unimodular (area 1/2) triangle with vertices in \(H\cap \mathbb{Z}^2\)
    \item There is a triplet of integral points \((x_1,y_1),(x_2,y_2),(x_3,y_3)\in H\cap\mathbb{Z}^2\) and integral multiplicities \(n_1,n_2,n_3\in \mathbb{Z}_{\geq 0}\) such that \(n_1(x_1,y_1)+n_2(x_2,y_2)+n_3(x_3,y_3)=(a,b)\) and \(n_1+n_2+n_3=n.\)
\end{enumerate}
\end{lemma}

\begin{proof}[Proof of Lemma~\ref{lemma:IDP-general}]
(i)\(\implies\)(ii): Observe that \(H\) is a lattice polyhedron which can be triangulated such that no triangle contains a lattice point within its interior (to see this, if any triangle contains an interior lattice point, connect that point to the three vertices. If any edge contains additional lattice points, subdivide along these points. This refinement must terminate). By Pick's theorem these triangles have area 1/2. 

(ii)\(\implies\) (iii): Suppose \(P\) lies within the unimodular triangle \(T\) with vertices \((x_1,y_1),(x_2,y_2),(x_3,y_3).\) Define
\[
W:= \begin{bmatrix}
    x_1 & x_2 & x_3 \\
    y_1 & y_2 & y_3 \\
    1 & 1 & 1
\end{bmatrix}
\]
The area of the triangle \(T\) is given by \(\frac{1}{2}|\det W|.\) Since \(T\) has area \(\frac{1}{2}\) by definition, then \(\det W=\pm 1,\) i.e. \(W\) is a unimodular matrix.

Since \(P=(\frac{a}{n},\frac{b}{n})\) lies within \(T,\) for some \(\lambda_1,\lambda_2,\lambda_3\in\mathbb{R}_{\geq 0},\) we can write
\[
W\cdot \begin{bmatrix}
    \lambda_1 \\ \lambda_2 \\ \lambda_3 
\end{bmatrix} = \begin{bmatrix}
    a/n \\ b/n \\ 1
\end{bmatrix} \implies W\cdot \begin{bmatrix}
    n\lambda_1 \\ n\lambda_2 \\ n\lambda_3
\end{bmatrix} = \begin{bmatrix}
    a \\ b \\ n
\end{bmatrix}.
\]
Now, since we know that \(W\) has integral entries, \(\det W=\pm 1,\) and \(a,b,n\) are integers, then by applying Cramer's Rule, \(n\lambda_1,n\lambda_2,n\lambda_3\) are integral. Then set \(n_1:=n\lambda_1,n_2:=n\lambda_2, n_3:=n\lambda_3.\)

(iii)\(\implies\) (i): By definition, each \((x_j,y_j)\in H\cap \mathbb{Z}^2,\) and \(\sum_{i=1}^n(x_j,y_j)=(a,b).\) Then, by convexity, their centroid \(P=(a/n,b/n)\) lies in \(H.\) 
\end{proof}

\begin{proof}[Proof of Lemma~\ref{lemma:IDP}]
    Since \(IH_i\) forms a lattice polyhedron, we may apply Lemma~\ref{lemma:IDP-general} by substituting \(P:=P_\mathcal{I},\) and \(H:=IH_i\) for any \(i\in N.\)
\end{proof}

The following helps us construct the unimodular triangle of Lemma~\ref{lemma:IDP-general} (ii) in polynomial time. It is a generalisation of Lemma~\ref{lemma:alg-time} which we use to prove Theorem~\ref{thm:mixed}.

\begin{lemma}[Theorem 6.1 of \cite{mccormick_polynomial_1997}]\label{lemma:alg-time-general}
Let \(R\subseteq\mathbb{R}^2\) be a (possibly unbounded) polyhedron defined by \(s\) inequalities, with largest coefficient \(A\), and let \(H\) be the integer hull of \(R\). For any \(P\in H,\) a unimodular triangle with vertices in \(H\cap\mathbb{Z}^2\) containing \(P\) can be found in time
\(        O\bigl(s\log A+(\log A)^2\bigr).
   \)
\end{lemma}

\section{MMS for mixed items} \label{sec:mixed}

We show that we can find in polynomial time MMS allocations for instances where there is only one type of objective good (positive utility to all agents), and one type of objective chore (negative utility to all agents).

For each agent, we use the regular valuation function \(v_i:M\to \mathbb{R}\) for agent \(i.\) An item \(o\) is a good (resp. chore) for agent \(i\) if \(v_i(o)>0\) \((v_i(o)<0).\) Recall the maximin share of an agent is defined by
\[
\MMS_i(M):=\max_{(X_i)_{i\in N}\in \Pi }\min_{j\in N}v_i(X_j),
\]
so an MMS allocation is one where each agents receives \textit{at least} their MMS value. We re-define MMS-feasibility to represent if the value of a bundle \(B\) for some agent \textit{exceeds} their MMS value.

When there is a fair division instance with one type of \textit{objective good} and one type of \textit{objective chore}, WLOG let Type A refer to the goods and Type B refer to the chores. Define \(v_i(o):=\alpha_i>0\) if \(o\) is a good, and \(v_i(o):=-\beta_i<0\) if \(o\) is a chore. Then let \(R_i\) refer to the nonnegative region lying \textit{below} the line \(\alpha_ix-\beta_iy=\MMS_i\) where \(\alpha_i,\beta_i>0.\) Lemma~\ref{lemma:IDP-general} is the generalisation of Lemma~\ref{lemma:IDP} (for the chores setting), which can be applied to mixed items.

Although the regions \(R_i\) may now be unbounded, by Lemma~\ref{lemma:alg-time-general} it is still possible to find a unimodular triangle with vertices in \(R_i\cap \mathbb{Z}^2\) containing \(P_\mathcal{I}\) in polynomial time. Note also that the proof of Observation~\ref{obs:P} still holds:

\obsP*
  \begin{proof}
Observe that given instance \(\mathcal{I}\), any MMS allocation for agent \(i\in N\) corresponds to \(n\) points \((x_1,y_1),\dots,(x_n,y_n)\in R_i\cap \mathbb{Z}^2\) such that \(\sum_{i=1}^n(x_i,y_i)=(a,b).\) Then, by convexity, their centroid \(P\) lies in \(IH_i.\)
\end{proof}

Finally, we retain the same notation and call the triangle computed by Lemma~\ref{lemma:alg-time-general} as the ``MMS triangle" for agent \(i,\) denoted by \(\triangle_i(\mathcal{I}).\) We denote by \(\mu_i(\mathcal{I})\) the corresponding multiplicities of its vertices.

Now, we are ready to prove the following result.

\mixed*

\begin{proof}
If \(n=1,\) then there is a single agent \(i\in N\) with \(P_\mathcal{I}\in IH_i.\) Then we may simply assign the bundle \(P_\mathcal{I}\) to agent \(i\) since it is MMS-feasible for them. Assume the theorem holds for all \(n'<n,\) and suppose \(n>1.\)

\paragraph{Partial ordering of agents.} Recall \(v_i\) is the value of \(y\) where \(\alpha_i\cdot x(P_\mathcal{I})+\beta_i\cdot y=\MMS_i\) and \(h_i\) is the value of \(x\) where \(\alpha_i\cdot x+\beta_i\cdot y(P_\mathcal{I})=\MMS_i.\) Define the partial order \(\preceq'\) on the agents \(N\) so that for \(i,j\in N,\)
\[
i\preceq' j \iff v_i\leq v_j\text{ and }h_i\geq h_j.
\]

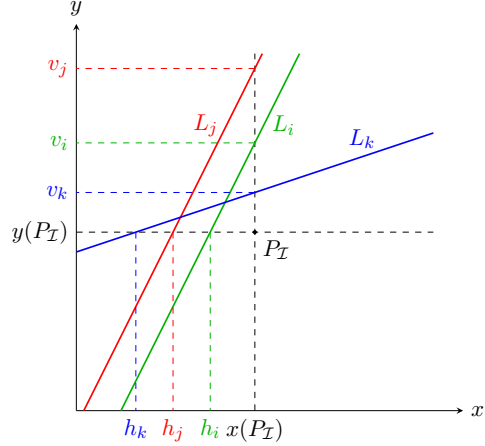
\begin{figure}[htbp]
    \centering
    \resizebox{0.4\linewidth}{!}{
    \begin{tikzpicture}[scale=1, >=stealth]
    \draw[->] (0,0) -- (6.5,0) node[right] {$x$};
    \draw[->] (0,0) -- (0,6.5) node[above] {$y$};

    \coordinate (Pnew) at (3,3);
    \fill (Pnew) circle (1pt) node[below right]{\(P_\mathcal{I}\)};

    \draw[dashed] (3,0) node[below]{\(x(P_\mathcal{I})\)} -- (3,6);
    \draw[dashed] (0,3) node[left]{\(y(P_\mathcal{I})\)} -- (6,3);
    
    \coordinate (I_j) at (0.8*6/2.8, -0.8*6/2.8 + 4);

    \draw[green!70!black, thick] (3/4,0) -- (15/4,6) node[pos=0.8, right, green!70!black] {\(L_i\)}; 
    
    \draw[dashed, green!70!black] (3,9/2) -- (0,9/2) node[pos=1, left, green!70!black] {\(v_i\)}; 
    
    \draw[dashed, green!70!black] (9/4,3) -- (9/4,0) node[pos=1, below, green!70!black] {\(h_i\)}; 
    
    \draw[red, thick] (1/8,0) -- (25/8,6) node[pos=0.8, left, red] {\(L_j\)}; 

    \draw[dashed, red] (3,23/4) -- (0,23/4) node[pos=1, left, red] {\(v_j\)}; 

    \draw[dashed, red] (13/8, 3) -- (13/8,0) node[pos=1, below, red] {\(h_j\)}; 

    \draw[blue, thick] (0,8/3) -- (6,14/3) node[pos=0.8, above, blue] {\(L_k\)}; 

    \draw[dashed, blue] (3,11/3) -- (0, 11/3) node[pos=1, left, blue] {\(v_k\)}; 

    \draw[dashed, blue] (1,3) -- (1,0) node[pos=1, below, blue] {\(h_k\)}; 

    \end{tikzpicture}
    }
    \caption{\(L_i\preceq' L_j.\) \((L_i,L_j)\) and \((L_j,L_k)\) are incomparable. If \(x > x(P_\mathcal{I}),\) then \(L_i\) is under \(L_j.\) If \(y<y(P_\mathcal{I}),\) then \(L_i\) is under \(L_k.\)}\label{fig:partial-order-mixed}
    \end{figure}

    Depending on the geometric configuration of vertices of \(\triangle_i(\mathcal{I})\) (i.e. whether they lie ``top-right"/``bottom-left"/etc. of \(P_\mathcal{I}\)) and their corresponding multiplicities. We assign these bundles to different agents case-by-case.

By Observation~\ref{obs:P}, note that we know the feasibility point \(P_\mathcal{I}\in \cap_{i=1}^nIH_i\) for all \(i\in N.\)

Let \(i\in N\) be a minimal agent under \(\preceq'.\) Now let \(A,B,C\) be the vertices of agent \(i\)'S MMS triangle \(\triangle_i(\mathcal{I}),\) and \(n_A,n_B,n_C\in \mu_i(\mathcal{I})\) their corresponding multiplicities. The procedure works recursively by allocating bundles \(A,B\) or \(C\) (possibly more than 1 of each) to satisfy some subset of agents, while ensuring that \(\MMS_i(\mathcal{I}')\geq \MMS_i(\mathcal{I})\) for all \(i\in N'\), where \(\mathcal{I}'\) is the remaining instance and \(N'\) is the remaining set of agents. We can then recurse by finding an MMS allocation for \(\mathcal{I}'.\)

In the allocation step, it suffices to show that after removing the corresponding bundles and agents, the new feasibility point, \(P_{\mathcal{I}'}\) remains in \(\cap_{i\in N'}IH_i.\) This is because Lemma~\ref{lemma:IDP-general} would imply that \(\MMS_i(\mathcal{I}')\geq \MMS_i(\mathcal{I}).\)

There are 5 cases in total (Case 1, and four subcases for Case 2). We prove the correctness of each case separately. \\

\noindent\textbf{Case 1:} Some \(X\in \{A,B,C\}\) lies ``top-left" of \(P_\mathcal{I}.\)

\begin{figure}[htbp]
    \centering
    \resizebox{0.4\linewidth}{!}{
    \begin{tikzpicture}[scale=1, >=stealth]
    
    \draw[->] (0,0) -- (6,0) node[right] {$x$};
    \draw[->] (0,0) -- (0,6) node[above] {$y$};
    
    \coordinate (P) at (3,4);
    \fill (P) circle (1pt) node[above right]{\(P_\mathcal{I}\)};
    
    \draw[dashed] (6,4) -- (0,4) node[left]{\(y(P_\mathcal{I})\)};
    \draw[dashed] (3,6) -- (3,0) node[below]{\(x(P_\mathcal{I})\)};
    
    \coordinate (X) at (1,5);
    \coordinate (P') at (4,3.5);
    \fill (X) circle (1pt) node[above left]{\(X\)};
    
    \draw[dashed] (X) -- (P');
    \fill (P') circle (1pt) node[right]{\(P_{\mathcal{I}'}\)};
    
    \coordinate (O) at (0,0);
    \fill (O) circle (1pt) node[below left]{\(O\)};
    
    \begin{scope}[on background layer]
        \fill[green!15] (3,4) -- (6,4) -- (6,0) -- (3,0) -- cycle;
    \end{scope}
    
    \end{tikzpicture}
    }
    \caption{(Case 1). Upon assigning bundle \(X\), the feasibility point moves from \(P_\mathcal{I}\) to \(P_{\mathcal{I}'}.\)} \label{fig:mixed-case-1-illustrative}
\end{figure}
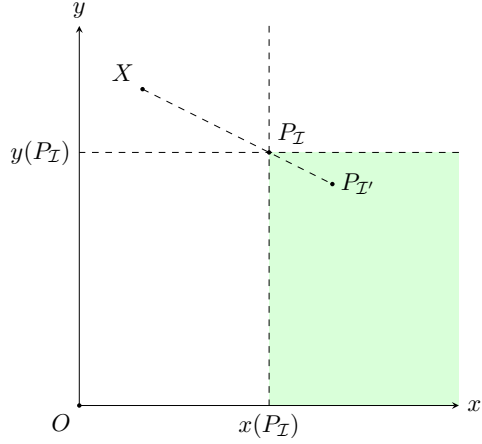

\emph{Allocation.} Assign \(X\) to agent \(i.\) The remaining instance is \(\mathcal{I}'=(a-x(X),b-y(X),N\setminus \{i\}).\)

\emph{Correctness.} See Fig~\ref{fig:mixed-case-1-illustrative}. Since \(X\in R_i\cap \mathbb{Z}^2,\) it is \(\MMS\)-feasible for agent \(i\) so we may assign them bundle \(X.\) For the remaining instance \(\mathcal{I}',\) the new feasibility point \(P_{\mathcal{I}'}=(\frac{a-x(X)}{n-1},\frac{b-y(X)}{n-1})\) will lie away from \(P_\mathcal{I}\) in the direction of \(\overrightarrow{XP_\mathcal{I}},\) which is ``bottom-right" of \(P_\mathcal{I}.\) Since \(a-x(X)\geq 0\) and \(b-y(X)\geq 0,\) \(P_{\mathcal{I}'}\) lies within the bottom-right region defined by the lines \(x=x(P_\mathcal{I}),y=y(P_\mathcal{I})\) and \(y=0.\) By convexity \(P_{\mathcal{I}'}\in \cap_{j\in N\setminus \{i\}} IH_j\) and we can induct. \\

    
    
    
    
    
    
    
    

\noindent\textbf{Case 2:} None of the vertices \(A,B,C\) of \(\triangle_i(\mathcal{I})\) lie ``top-left" of \(P_\mathcal{I}.\) 

Now, observe that at least one of \(A,B,C\) must have \(y\)-coordinate which is greater than or equal to \(y(P_\mathcal{I})\) (since the weighted centroid of \(A,B,C\) is \(P_\mathcal{I}\)). WLOG, let this be \(B.\) If we are not in Case 1,  then we must have \(x(B)>x(P_\mathcal{I}),\) i.e. \(B\) is ``top-right" of \(P_\mathcal{I}.\) Similarly, we may assume that \(C\) lies ``bottom-left" of \(P_\mathcal{I}\) (i.e \(y(C)\leq y(P_\mathcal{I})\) and \(x(C)<x(P_\mathcal{I})\)), and WLOG that \(A\) lies under \(BC\) (so \(A,B,C\) are in anti-clockwise order).

Next, define the set \(S:=\{j\in N\setminus \{i\}:\frac{\beta_j}{\alpha_j}\leq \frac{\beta_i}{\alpha_i}\}\) to be the set of agents \(j\in N\) such that the slope of the line \(L_j\) is steeper than that of \(L_i.\) Analogously, we define \(G:=\{j\in N\setminus \{i\}:\frac{\beta_j}{\alpha_j}> \frac{\beta_i}{\alpha_i}\}\) to be the set of agents \(j\in N\) such that the slope of the line \(L_j\) is gentler than that of \(L_i.\) 

Let \(I_{j}\) for \(j\in N\) denote the intersection point of lines \(L_i\) and \(L_j.\) For \(j\in N\), let \(R_j^\text{right}:=\{(x,y)\in R_j: x\geq x(P_\mathcal{I})\}\) and \(R_j^\text{below}:=\{(x,y)\in R_j: y\leq y(P_\mathcal{I})\}.\) We can make the following observation:

\begin{observation}\label{obs:mixed-steep-gentle}
    For all \(j\in S, R_i^\text{right}\subseteq R_j^\text{right},\) and for all \(k\in G, R_i^\text{below}\subseteq R_k^\text{below}\) (see  Fig~\ref{fig:mixed-partial-order-subset}).
\end{observation}
\begin{figure}[htbp]
    \centering
    \begin{subfigure}[t]{0.48\textwidth}
        \centering
        \resizebox{0.8\linewidth}{!}{
        \begin{tikzpicture}[scale=1, >=stealth]

        \draw[->] (0,0) -- (6.5,0) node[right] {$x$};
        \draw[->] (0,0) -- (0,6.5) node[above] {$y$};
        
        \coordinate (P) at (4,1);
        \fill (P) circle (1pt) node[below right]{\(P_\mathcal{I}\)};
        
        \coordinate (I_j) at (3.5,2.5);
        \fill (I_j) circle (1pt) node[above left]{\(I_j\)};
        
        \draw[dashed] (4,0) node[below]{\(x(P)\)} -- (4,6);
        \draw[dashed] (0,1) node[left]{\(y(P)\)} -- (6,1);
        
        \draw[green!70!black, thick] (1,0) -- (6,5) node[pos=0.8, right, green!70!black] {\(L_i\)}; 
        
        \draw[dashed, green!70!black] (0,3) -- (4,3) node[pos=0, left, green!70!black] {\(v_i\)}; 
        
        \draw[dashed, green!70!black] (2,1) -- (2,0) node[pos=1, below, green!70!black] {\(h_i\)}; 
        
        \draw[red, thick] (8/3,0) -- (14/3,6) node[pos=0.8, right, red] {\(L_j\)}; 
        
        \draw[dashed, red] (0,4) -- (4,4) node[pos=0, left, red] {\(v_j\)}; 
        
        \draw[dashed, red] (3,1) -- (3,0) node[pos=1, below, red] {\(h_j\)}; 
        
        \begin{scope}[on background layer]
            \fill[green!15] (4,0) -- (4,3) -- (6,5) -- (6,0) -- cycle;
        \end{scope}
        
        \begin{scope}[on background layer]
            \fill[red!15] (4,3) -- (6,5) -- (6,6) -- (14/3,6) -- (4,4) -- cycle;
        \end{scope}
        
        \end{tikzpicture}
}
    \end{subfigure}
    \hfill
    \begin{subfigure}[t]{0.48\textwidth}
        \centering
        \resizebox{0.8\linewidth}{!}{
        \begin{tikzpicture}[scale=1, >=stealth]
        
        \draw[->] (0,0) -- (6.5,0) node[right] {$x$};
        \draw[->] (0,0) -- (0,6.5) node[above] {$y$};
        
        \coordinate (P) at (4,1);
        \fill (P) circle (1pt) node[below right]{\(P_\mathcal{I}\)};
        
        \draw[dashed] (4,0) node[below]{\(x(P)\)} -- (4,6);
        \draw[dashed] (0,1) node[left]{\(y(P)\)} -- (6,1);
        
        \draw[green!70!black, thick] (1,0) -- (6,5) node[pos=0.8, right, green!70!black] {\(L_i\)}; 
        
        \draw[dashed, green!70!black] (0,3) -- (4,3) node[pos=0, left, green!70!black] {\(v_i\)}; 
        
        \draw[dashed, green!70!black] (2,1) -- (2,0) node[pos=1, below, green!70!black] {\(h_i\)}; 
        
        \draw[blue, thick] (0,0.7) -- (6,0.8*6+0.7) node[pos=0.8, above, blue] {\(L_k\)}; 
        
        \draw[dashed, blue] (0,0.8*4+0.7) -- (4,0.8*4+0.7) node[pos=0, left, blue] {\(v_k\)}; 
        
        \draw[dashed, blue] (0.3/0.8,1) -- (0.3/0.8,0) node[pos=1, below, blue] {\(h_k\)}; 
        
        \begin{scope}[on background layer]
            \fill[green!15] (1,0) -- (2,1) -- (6,1) -- (6,0) -- cycle;
        \end{scope}
        
        \begin{scope}[on background layer]
            \fill[blue!15] (1,0) -- (0,0) -- (0,0.7) -- (0.3/0.8,1) -- (2,1) -- cycle;
        \end{scope}
        
        \end{tikzpicture}
        }
    \end{subfigure}
    \caption{Illustration of Obs~\ref{obs:mixed-steep-gentle}. \(R_i^\text{right}\subseteq R_j^\text{right}\) for \(j\in S\) (left) and \(R_i^\text{below}\subseteq R_k^\text{below}\) for \(k\in G\).}
    \label{fig:mixed-partial-order-subset}
\end{figure}
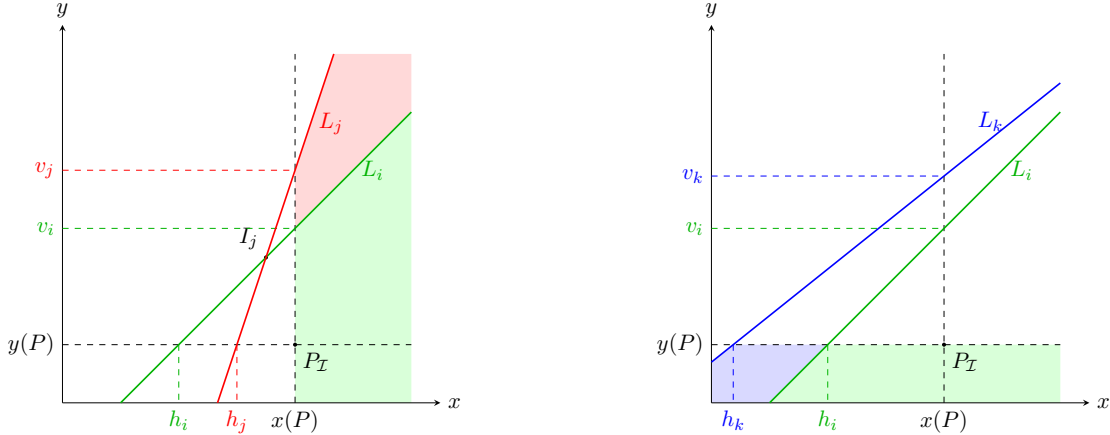
\begin{proof}
    Let \(j\in S.\) We show \(R_i^\text{right}\subseteq R_j^\text{right}\). Suppose \(x(I_j)>x(P_\mathcal{I}).\) Then, since \(L_j\) is steeper than \(L_i,\) we must have \(v_j<v_i\) and \(h_j>h_i,\) but this contradicts the minimality of \(i\) under \(\preceq'.\) Thus \(x(I_j)\leq x(P_\mathcal{I}).\) Then suppose some point \(R\in R_i^{\text{right}}.\) Since \(R\) lies under \(L_i\), \(L_j\) is steeper than \(L_i\), and \(x(R)\geq x(P_\mathcal{I})\geq x(I_j),\) then \(R\in R_j^{\text{right}}.\) The argument for \(R_i^\text{below}\subseteq R_k^\text{below}\) for any \(k\in G\) is symmetric.
\end{proof}

Note that this may not necessarily imply that when \(x\leq x(P), IH_i\subseteq IH_j,\) or when \(y\leq y(P), IH_i\subseteq IH_k\).

Now, we use the above fact to handle the remaining cases. \\

\noindent\textbf{Case 2a:} \(n_B\geq |S|\)

\begin{figure}[htbp]
    \centering
    \begin{subfigure}[t]{0.48\textwidth}
        \centering
        \resizebox{0.8\linewidth}{!}{
        \begin{tikzpicture}[scale=1, >=stealth]
        
        \draw[->] (0,0) -- (6,0) node[right] {$x$};
        \draw[->] (0,0) -- (0,6) node[above] {$y$};
        
        \coordinate (P) at (3.03,2.2);
        \fill[red] (P) circle (1pt) node[right] {\tiny\(P_{\mathcal{I}}\)};
        
        \draw[green!70!black, thick] (3/2,0) -- (9/2,6) node[pos=0.8, left, green!70!black] {\(L_i\)}; 
        
        \draw[blue, thick] (0,1) -- (6,0.7*6+1) node[pos=0.2, above, blue] {\(L_k\)}; 
        
        \coordinate (A) at (3,1);
        \coordinate (B) at (4.5,5.5);
        \coordinate (C) at (1.7,0.4);
        
        \draw[line width=0.6pt, rounded corners=1pt] (A) -- (B) -- (C) -- cycle;
        \fill (A) circle (1pt) node[right] {\(A\)};
        \fill (B) circle (1pt) node[above] {\(B\)};
        \fill (C) circle (1pt) node[left] {\(C\)};
        
        \coordinate (Z) at (2.3097,0.5929); 
        \fill (Z) circle (1pt) node[right] {\tiny \(Z\)};
        \draw[dashed, line width=0.6pt] (B) -- (Z);
        
        \end{tikzpicture}
        }
        \caption{\(y(A)\leq y(P_\mathcal{I}).\) Red point is \(P_\mathcal{I}.\) \(Z\) lies on \(\overline{AC}\) and \(P_{\mathcal{I}'}\) lies on \(\overline{P_{\mathcal{I}}Z}\).}
    \end{subfigure} \hfill
    \begin{subfigure}[t]{0.48\textwidth}
        \centering
        \resizebox{0.8\linewidth}{!}{
        \begin{tikzpicture}[scale=1, >=stealth]

        \draw[->] (0,0) -- (6,0) node[right] {$x$};
        \draw[->] (0,0) -- (0,6) node[above] {$y$};
        
        \coordinate (P) at (3,2.2);
        \fill[red] (P) circle (1pt) node[right] {\tiny\(P_{\mathcal{I}}\)};
        
        \draw[green!70!black, thick] (3/2,0) -- (9/2,6) node[pos=0.8, left, green!70!black] {\(L_i\)}; 
        
        \draw[blue, thick] (0,1) -- (6,0.7*6+1) node[pos=0.2, above, blue] {\(L_k\)}; 
        
        \coordinate (A) at (6,5.5);
        \coordinate (B) at (4.5,5.5);
        \coordinate (C) at (1.7,0.4);
        
        \draw[line width=0.6pt, rounded corners=1pt] (A) -- (B) -- (C) -- cycle;
        \fill (A) circle (1pt) node[right] {\(A\)};
        \fill (B) circle (1pt) node[above] {\(B\)};
        \fill (C) circle (1pt) node[left] {\(C\)};
        
        \coordinate (Z) at (2.7311,1.6236); 
        \fill (Z) circle (1pt) node[right] {\tiny \(Z\)};
        \draw[dashed, line width=0.6pt] (B) -- (Z);
        
        \begin{scope}[on background layer]
            \fill[green!15] (P) -- (3,0) -- (1.7,0) -- (C) -- cycle;
        \end{scope}

        \draw[dashed, line width=0.6pt] (P) -- (3,0) node[below] {\tiny\(x(P_\mathcal{I})\)};
        \draw[dashed, line width=0.6pt] (P) -- (C);
        \draw[dashed, line width=0.6pt] (C) -- (1.7,0) node[below] {\tiny\(x(C)\)};
        
        \end{tikzpicture}
        }
        \caption{\(y(A)>y(P_\mathcal{I}).\) Red point is \(P_\mathcal{I}\). \(P_{\mathcal{I}'}\) lies within region bounded by \((x(P_\mathcal{I}),0),P_\mathcal{I},C,(x(C),0).\))}
    \end{subfigure}
    \caption{Illustration of Case 2a.}
    \label{fig:mixed-case-2a}
\end{figure}
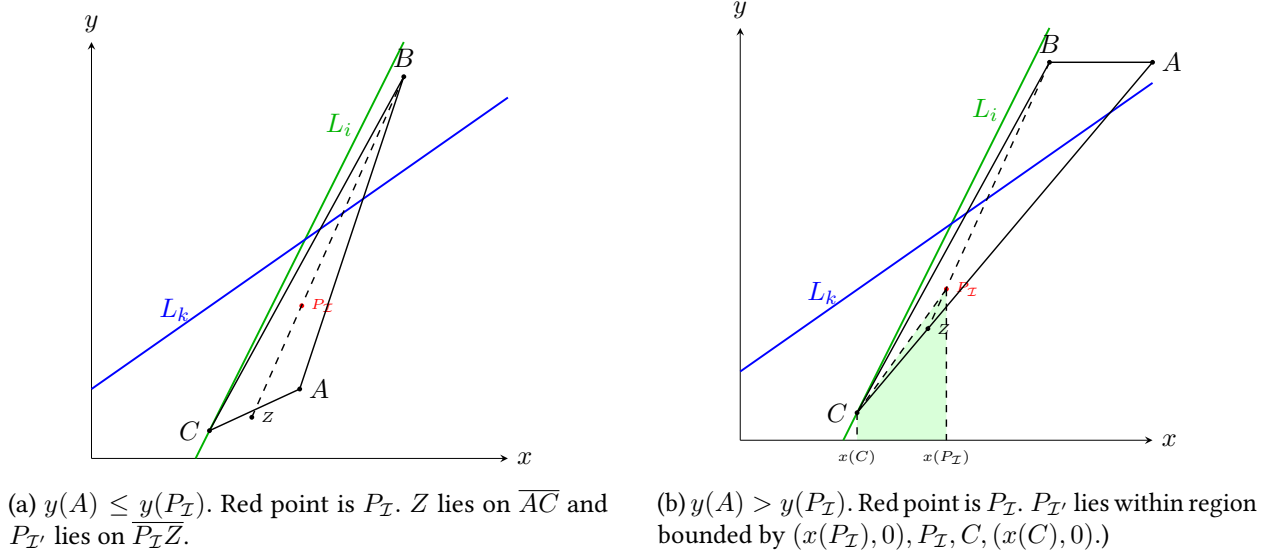

\emph{Allocation.} Assign each agent in \(S\) the bundle \(B.\) The remaining instance is \(\mathcal{I}'=(a-|S|\cdot x(B),b-|S|\cdot y(B),N\setminus S).\) 

\emph{Correctness.} See Fig~\ref{fig:mixed-case-2a}. By Observation~\ref{obs:mixed-steep-gentle}, the bundles are \(\MMS\)-feasible for all \(j\in S.\) Let the remaining instance be \(\mathcal{I}'=(a-|S|\cdot x(B),b-|S|\cdot y(B),N\setminus S)\) and note that \(N\setminus S = G\cup \{i\}.\) The new feasibility point \(P_{\mathcal{I}'}\) must lie along the segment \(\overline{P_\mathcal{I} Z}\) where \(Z = BP_\mathcal{I} \cap AC\) and ``bottom-left" of \(P_\mathcal{I}\). Now since \(P_{\mathcal{I}'}\) is still a weighted centroid of \(A,B,C\in IH_i,\) then \(P_{\mathcal{I}'}\in IH_i.\) 

\noindent Fix \(k\in G.\) By assumption, since \(y(C)\leq y(P_\mathcal{I}),\) by Observation~\ref{obs:mixed-steep-gentle}, we have \(C\in R_i^\text{below}\subseteq R_k^\text{below}\subseteq R_k.\) Since \(C\in\mathbb{Z}^2,\) then \(C\in IH_k.\) 

\noindent Suppose \(y(A) \leq y(P_\mathcal{I})\). Then by the same argument \(A\in IH_k.\) Hence since \(Z\) lies along \(\overline{AC}\), by convexity, \(Z\in IH_k\). 

\noindent Otherwise, suppose \(y(A) > y(P_\mathcal{I})\). Then observe that \(Z\) lies within the polygon defined by the points \((x(P_\mathcal{I}),0),P_\mathcal{I},C,(x(C),0).\) Since \(C\in IH_k,\) then \((x(C),0)\in IH_k.\) Also, by assumption, \(P_\mathcal{I}\in IH_k\) and hence \((x(P_\mathcal{I}),0)\in IH_k.\) Thus by convexity \(Z\in IH_k.\)

\noindent In either case, we have \(Z\in \cap_{k\in N\setminus S}IH_k\), hence \(P_{\mathcal{I}'}\in \cap_{k\in N\setminus S}IH_k\) by the fact that \(Z, P_{\mathcal{I}}\in \cap_{k\in N\setminus S}IH_k\) and the convexity of \(\cap_{k\in N\setminus S}IH_k\), and we may induct. \\

\newpage
\noindent\textbf{Case 2b:} \(n_C\geq |G|.\)

\begin{figure}[htbp]
    \centering
    \begin{subfigure}[t]{0.48\textwidth}
        \centering
        \resizebox{0.8\linewidth}{!}{
        \begin{tikzpicture}[scale=1, >=stealth]
        
        \draw[->] (0,0) -- (6,0) node[right] {$x$};
        \draw[->] (0,0) -- (0,6) node[above] {$y$};
        
        \coordinate (P) at (3.7,2.7);
        \fill[red] (P) circle (1pt) node[right] {\tiny\(P_{\mathcal{I}}\)};
        
        \draw[green!70!black, thick] (0,1) -- (6,0.6*6+1) node[pos=0.2, above, green!70!black] {\(L_i\)}; 
        
        \draw[red, thick] (2,0) -- (5,6) node[pos=0.8, left, red] {\(L_j\)}; 
        
        \coordinate (A) at (6,3);
        \coordinate (B) at (5,4);
        \coordinate (C) at (2,1.5);
        
        \draw[line width=0.6pt, rounded corners=1pt] (A) -- (B) -- (C) -- cycle;
        \fill (A) circle (1pt) node[right] {\(A\)};
        \fill (B) circle (1pt) node[above] {\(B\)};
        \fill (C) circle (1pt) node[left] {\(C\)};
        
        \coordinate (Z) at (5.22414,3.77586); 
        \fill (Z) circle (1pt) node[right] {\tiny \(Z\)};
        \draw[dashed, line width=0.6pt] (C) -- (Z);
        
        \end{tikzpicture}
        }
        \caption{\(y(A)\geq y(P_\mathcal{I}).\) Red point is \(P_\mathcal{I}.\) \(Z\) lies on \(\overline{AB}\) and \(P_{\mathcal{I}'}\) lies on \(\overline{P_{\mathcal{I}}Z}\).}
    \end{subfigure} \hfill
    \begin{subfigure}[t]{0.48\textwidth}
        \centering
        \resizebox{0.8\linewidth}{!}{
        \begin{tikzpicture}[scale=1, >=stealth]

        \draw[->] (0,0) -- (6,0) node[right] {$x$};
        \draw[->] (0,0) -- (0,6) node[above] {$y$};
        
        \coordinate (P) at (3.7,2.7);
        \fill[red] (P) circle (1pt) node[below] {\tiny\(P_{\mathcal{I}}\)};
        
        \draw[green!70!black, thick] (0,1) -- (6,0.6*6+1) node[pos=0.2, above, green!70!black] {\(L_i\)}; 
        
        \draw[red, thick] (2,0) -- (5,6) node[pos=0.8, left, red] {\(L_j\)}; 
        
        \coordinate (A) at (2,0.2);
        \coordinate (B) at (5,4);
        \coordinate (C) at (2,1.5);
        
        \draw[line width=0.6pt, rounded corners=1pt] (A) -- (B) -- (C) -- cycle;
        \fill (A) circle (1pt) node[right] {\(A\)};
        \fill (B) circle (1pt) node[above] {\(B\)};
        \fill (C) circle (1pt) node[left] {\(C\)};
        
        \coordinate (Z) at (4.31818,3.15035); 
        \fill (Z) circle (1pt) node[right] {\tiny \(Z\)};
        \draw[dashed, line width=0.6pt] (C) -- (Z);
        
        \begin{scope}[on background layer]
            \fill[green!15] (P) -- (3.7,0) -- (5,0) -- (B) -- cycle;
        \end{scope}

        \draw[dashed, line width=0.6pt] (P) -- (3.7,0) node[below] {\tiny\(x(P_\mathcal{I})\)};
        \draw[dashed, line width=0.6pt] (P) -- (B);
        \draw[dashed, line width=0.6pt] (B) -- (5,0) node[below] {\tiny\(x(B)\)};
        
        \end{tikzpicture}
        }
        \caption{\(y(A)<y(P_\mathcal{I}).\) Red point is \(P_\mathcal{I}\). \(P_{\mathcal{I}'}\) lies within region bounded by \((x(P_\mathcal{I}),0),P_\mathcal{I},B,(x(B),0).\)}
    \end{subfigure}
    \caption{Illustration of Case 2b.}
    \label{fig:mixed-case-2b}
\end{figure}

\emph{Allocation.} Assign each agent in \(G\) the bundle \(C.\) The remaining instance is \(\mathcal{I}'=(a-|G|\cdot x(C),b-|G|\cdot y(C),N\setminus G).\) 

\emph{Correctness.} See Fig~\ref{fig:mixed-case-2b}. By Observation~\ref{obs:mixed-steep-gentle}, the bundles are \(\MMS\)-feasible for all \(j\in G.\) Let the remaining instance be \(\mathcal{I}'=(a-|G|\cdot x(C),b-|G|\cdot y(C),N\setminus G)\) and note that \(N\setminus G = S\cup \{i\}.\) The new feasibility point \(P_{\mathcal{I}'}\) must lie along the segment \(\overline{P_\mathcal{I} Z}\) where \(Z = CP_\mathcal{I} \cap AB\) and ``top-right" of \(P_\mathcal{I}\). Now since \(P_{\mathcal{I}'}\) is still a weighted centroid of \(A,B,C\in IH_i,\) then \(P_{\mathcal{I}'}\in IH_i.\) 

\noindent Fix \(j\in S.\) By assumption, since \(y(B)\geq y(P_\mathcal{I}),\) by Observation~\ref{obs:mixed-steep-gentle}, we have \(B\in R_i^\text{right}\subseteq R_j^\text{right}\subseteq R_j.\) Since \(B\in\mathbb{Z}^2,\) then \(B\in IH_j.\) 

\noindent Suppose \(y(A) \geq y(P_\mathcal{I})\). Then by the same argument \(A\in IH_j.\) Hence since \(Z\) lies along \(\overline{AC}\), by convexity, \(Z\in IH_j\). 

\noindent Otherwise, suppose \(y(A) < y(P_\mathcal{I})\). Then observe that \(Z\) lies within the region bounded by the points \((x(P_\mathcal{I}),0),P_\mathcal{I},B,(x(B),0).\) 

\((x(P_\mathcal{I}),0),P_\mathcal{I},C,(x(C),0).\) Since \(B\in IH_j,\) then \((x(B),0)\in IH_j.\) Also, by assumption, \(P_\mathcal{I}\in IH_j\) and hence \((x(P_\mathcal{I}),0)\in IH_j.\) Thus by convexity \(Z\in IH_j.\)

\noindent In either case, we have \(Z\in \cap_{j\in N\setminus G}IH_j\), hence \(P_{\mathcal{I}'}\in \cap_{j\in N\setminus G}IH_j\) by the fact that \(Z, P_{\mathcal{I}}\in \cap_{j\in N\setminus G}IH_j\) and the convexity of \(\cap_{j\in N\setminus G}IH_j\), and we may induct. \\

\noindent\textbf{Case 2c:} \(n_B<|S|, n_C<|G|\) and \(A\) lies ``bottom-right" of \(P_\mathcal{I}\) 

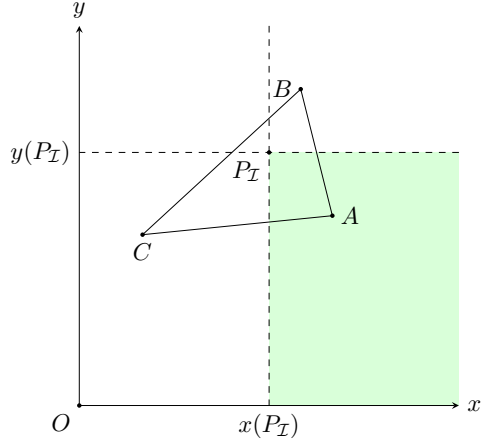
\begin{figure}[htbp]
    \centering
    \resizebox{0.4\linewidth}{!}{
    \begin{tikzpicture}[scale=1, >=stealth]
    
    \draw[->] (0,0) -- (6,0) node[right] {$x$};
    \draw[->] (0,0) -- (0,6) node[above] {$y$};
    
    \coordinate (P) at (3,4);
    \fill (P) circle (1pt) node[below left]{\(P_\mathcal{I}\)};
    
    \draw[dashed] (6,4) -- (0,4) node[left]{\(y(P_\mathcal{I})\)};
    \draw[dashed] (3,6) -- (3,0) node[below]{\(x(P_\mathcal{I})\)};
    
    \coordinate (O) at (0,0);
    \fill (O) circle (1pt) node[below left]{\(O\)};
    
    \begin{scope}[on background layer]
        \fill[green!15] (P) -- (6,4) -- (6,0) -- (3,0) -- cycle;
    \end{scope}
    
    \coordinate (B) at (3.5,5);
    \fill (B) circle (1pt) node[left] {\(B\)};
    
    \coordinate (A) at (4,3);
    \fill (A) circle (1pt) node[right] {\(A\)};
    
    \coordinate (C) at (1,2.7);
    \fill (C) circle (1pt) node[below] {\(C\)};
    
    \draw (A) -- (B) -- (C) -- cycle;
    
    \end{tikzpicture}
    }
    \caption{(Case 2c) \(A\) lies ``bottom-right" of \(P_\mathcal{I}\) hence is MMS-feasible for all \(i\in N.\)} \label{fig:mixed-case-2c}
\end{figure}

\emph{Allocation.} Assign all \(n_B\) bundles \(B\) to agents in \(S,\) and all \(n_C\) bundles \(C\) to agents in \(G.\)

\emph{Correctness.} See Fig~\ref{fig:mixed-case-2c}. By Observation~\ref{obs:mixed-steep-gentle}, these are MMS-feasible assignments. 

Now since \(A\) is ``bottom-right" of \(P_\mathcal{I},\) and \(P_\mathcal{I}\) lies under \(L_j\) for all \(j\in N,\) i.e. \(A\) is \(\MMS\)-feasible for all agents \(j\in N\). Hence we may simply assign \(A\) to any remaining agent, and since \(n_A\cdot A + n_B\cdot B + n_C\cdot C=(a,b)\) (by definition), this is a complete allocation and an \(\MMS\)-feasible one. \\

\newpage
\noindent\textbf{Case 2d:} \(n_B<|S|, n_C<|G|\) and \(A\) lies ``top-right" of \(P_\mathcal{I}.\)

\begin{figure}[htbp]
    \centering
    \resizebox{0.4\linewidth}{!}{
        \begin{tikzpicture}[scale=1, >=stealth]

        \draw[->] (0,0) -- (6,0) node[right] {$x$};
        \draw[->] (0,0) -- (0,6) node[above] {$y$};
        
        \coordinate (P) at (3,2.2);
        \fill[red] (P) circle (1pt) node[right] {\tiny\(P_{\mathcal{I}}\)};
        
        \draw[green!70!black, thick] (3/2,0) -- (9/2,6) node[pos=0.8, left, green!70!black] {\(L_i\)}; 
        
        \draw[blue, thick] (0,1) -- (6,0.7*6+1) node[pos=0.2, above, blue] {\(L_k\)}; 
        
        \coordinate (A) at (6,5.5);
        \coordinate (B) at (4.5,5.5);
        \coordinate (C) at (1.7,0.4);
        
        \draw[line width=0.6pt, rounded corners=1pt] (A) -- (B) -- (C) -- cycle;
        \fill (A) circle (1pt) node[right] {\(A\)};
        \fill (B) circle (1pt) node[above] {\(B\)};
        \fill (C) circle (1pt) node[left] {\(C\)};
        
        \coordinate (Z) at (2.7311,1.6236); 
        \fill (Z) circle (1pt) node[right] {\tiny \(Z\)};
        \draw[dashed, line width=0.6pt] (B) -- (Z);
        
        \begin{scope}[on background layer]
            \fill[green!15] (P) -- (3,0) -- (1.7,0) -- (C) -- cycle;
        \end{scope}

        \draw[dashed, line width=0.6pt] (P) -- (3,0) node[below] {\tiny\(x(P_\mathcal{I})\)};
        \draw[dashed, line width=0.6pt] (P) -- (C);
        \draw[dashed, line width=0.6pt] (C) -- (1.7,0) node[below] {\tiny\(x(C)\)};
        
        \end{tikzpicture}
        }
        \caption{(Case 2d) - same diagram as Fig~\ref{fig:mixed-case-2a} (b), but different procedure. \(P_{\mathcal{I}'}\) now lies along \(\overline{ZC}.\)} \label{fig:mixed-case-2d}
\end{figure}
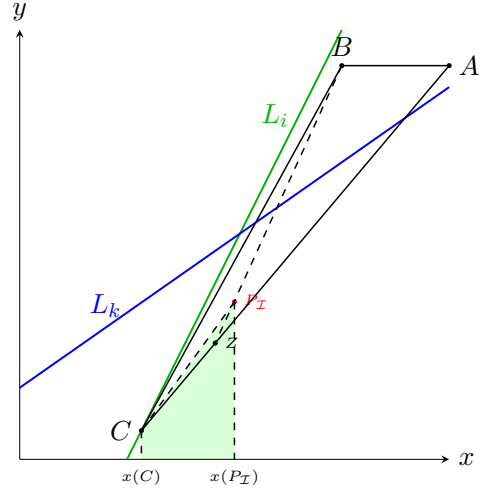

\emph{Allocation.} Assign \(n_B\) agents in \(S\) the bundle \(B.\) Assign \(|S|-n_B\) remaining agents in \(S\) the bundle \(A.\) The remaining instance is \(\mathcal{I}'=(a-n_B\cdot x(B)-(|S|-n_B)\cdot x(A),b-n_B\cdot y(B)-(|S|-n_B)\cdot y(A),N\setminus S)\)

\emph{Correctness.} See Fig~\ref{fig:mixed-case-2d}. Note that there are enough bundles since \(n_A+n_B+n_C=|S|+|G|+1\) and \(n_C<|G|\) implies \(n_A+n_B>|S|.\) By Observation~\ref{obs:mixed-steep-gentle}, these are MMS-feasible for all \(j\in S.\) Let the remaining instance be \(\mathcal{I}'=(a-n_B\cdot x(B)-(|S|-n_B)\cdot x(A),b-n_B\cdot y(B)-(|S|-n_B)\cdot y(A),N\setminus S)\) and note that \(N\setminus S=G\cup \{i\}.\) Let \(Z:=\frac{n_A\cdot A+n_C\cdot C}{n_A+n_C}\) be the weighted centroid of the remaining points after removing bundles \(B.\) This must lie along the segment \(\overline{AC},\) and is ``bottom-left" of \(P_\mathcal{I}\). Now upon removing \(|S|-n_B\) copies of \(A,\) the new feasibility point \(P_{\mathcal{I}'}\) must lie along the segment \(\overline{ZC},\) and is still ``bottom-left" of \(P_\mathcal{I}.\) Thus, like in Case 2a, this is within the polygon defined by the points \((x(P_\mathcal{I}),0),P_\mathcal{I},C,(x(C),0).\) Then by the same argument \(P_{\mathcal{I}'}\in \cap_{k\in N\setminus S} IH_k(D_k)\) and we may induct. \\

\noindent\textbf{Case 2e:} \(n_B<|S|, n_C<|G|\) and \(A\) lies ``bottom-left" of \(P_\mathcal{I}.\)

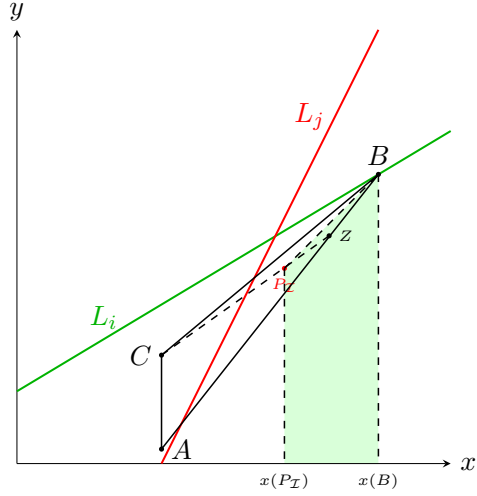
\begin{figure}[htbp]
    \centering
        \resizebox{0.4\linewidth}{!}{
        \begin{tikzpicture}[scale=1, >=stealth]

        \draw[->] (0,0) -- (6,0) node[right] {$x$};
        \draw[->] (0,0) -- (0,6) node[above] {$y$};
        
        \coordinate (P) at (3.7,2.7);
        \fill[red] (P) circle (1pt) node[below] {\tiny\(P_{\mathcal{I}}\)};
        
        \draw[green!70!black, thick] (0,1) -- (6,0.6*6+1) node[pos=0.2, above, green!70!black] {\(L_i\)}; 
        
        \draw[red, thick] (2,0) -- (5,6) node[pos=0.8, left, red] {\(L_j\)}; 
        
        \coordinate (A) at (2,0.2);
        \coordinate (B) at (5,4);
        \coordinate (C) at (2,1.5);
        
        \draw[line width=0.6pt, rounded corners=1pt] (A) -- (B) -- (C) -- cycle;
        \fill (A) circle (1pt) node[right] {\(A\)};
        \fill (B) circle (1pt) node[above] {\(B\)};
        \fill (C) circle (1pt) node[left] {\(C\)};
        
        \coordinate (Z) at (4.31818,3.15035); 
        \fill (Z) circle (1pt) node[right] {\tiny \(Z\)};
        \draw[dashed, line width=0.6pt] (C) -- (Z);
        
        \begin{scope}[on background layer]
            \fill[green!15] (P) -- (3.7,0) -- (5,0) -- (B) -- cycle;
        \end{scope}

        \draw[dashed, line width=0.6pt] (P) -- (3.7,0) node[below] {\tiny\(x(P_\mathcal{I})\)};
        \draw[dashed, line width=0.6pt] (P) -- (B);
        \draw[dashed, line width=0.6pt] (B) -- (5,0) node[below] {\tiny\(x(B)\)};
        
        \end{tikzpicture}
        }
    \caption{(Case 2e) - same diagram as Fig~\ref{fig:mixed-case-2b} (b), but different procedure. \(P_{\mathcal{I}'}\) now lies along \(\overline{ZB}.\)}\label{fig:mixed-case-2e}
\end{figure}

\emph{Allocation.} Assign \(n_C\) agents in \(G\) the bundle \(A.\) Assign \(|G|-n_C\) remaining agents in \(G\) the bundle \(A.\) The remaining instance is \(\mathcal{I}'=(a-n_C\cdot x(C)-(|G|-n_C)\cdot x(A),b-n_C\cdot y(C)-(|G|-n_C)\cdot y(A),N\setminus G)\)

\emph{Correctness.} See Fig~\ref{fig:mixed-case-2e}. Note that there are enough bundles since \(n_A+n_B+n_C=|S|+|G|+1\) and \(n_B<|S|\) implies \(n_A+n_C>|G|.\) By Observation~\ref{obs:mixed-steep-gentle}, these are MMS-feasible for all \(k\in S.\) Let the remaining instance be \(\mathcal{I}'=(a-n_C\cdot x(C)-(|G|-n_C)\cdot x(A),b-n_C\cdot y(C)-(|G|-n_C)\cdot y(A),N\setminus G)\) and note that \(N\setminus G=S\cup \{i\}.\) Let \(Z:=\frac{n_A\cdot A+n_B\cdot B}{n_A+n_B}\) be the weighted centroid of the remaining points after removing bundles \(C.\) This must lie along the segment \(\overline{AC},\) and is ``top-right" of \(P_\mathcal{I}\). Now upon removing \(|G|-n_C\) copies of \(A,\) the new feasibility point \(P_{\mathcal{I}'}\) must lie along the segment \(\overline{ZB},\) and is still ``top-right" of \(P_\mathcal{I}.\) Thus, like in Case 2b, this is within the polygon defined by the points \((x(P_\mathcal{I}),0),P_\mathcal{I},B,(x(B),0).\) Then by the same argument \(P_{\mathcal{I}'}\in \cap_{j\in N\setminus S} IH_j(D_j)\) and we may induct. \\

\end{proof}

\textbf{Time complexity.} By an application of Lemma~\ref{lemma:alg-time-general} (see Theorem 6.1 in \cite{mccormick_polynomial_1997}), the values \((\MMS_j)_{j\in N}\) may be computed in polynomial time. In each iteration, we may find the minimal agent \(i\in N\) under \(\preceq\) in \(O(n)\) time. Computing \(\triangle_i(\mathcal{I})\) also takes \(O(\log^2 \MMS_i)\) time by Lemma~\ref{lemma:alg-time-general}, and the number of agents to consider after each iteration decreases by at least one. Thus the allocation procedure can be done in polynomial time.



\end{document}